%% file: ICCsingle.tex
\documentclass[journal,onecolumn]{IEEEtran}
\usepackage{amsmath}
\usepackage{amsthm}
\usepackage{amsfonts}
\usepackage{amssymb}
\usepackage{bm}
\usepackage{microtype}
\usepackage{mathrsfs}
\usepackage{mathtools}
\usepackage{graphicx}
\usepackage{cite}
\usepackage{epstopdf}
\usepackage{bbm}
\DeclareMathOperator*{\argmin}{arg\,min}
\DeclareMathOperator*{\argmax}{arg\,max}

\newtheorem{lemma}{Lemma}
\newtheorem{theorem}{Theorem}

\newtheorem{assumption}{Assumption}
\newtheorem{example}{Example}
\newtheorem{remark}{Remark}
\usepackage{algorithm}
\usepackage{algpseudocode}

\algrenewcommand\alglinenumber[1]{\scriptsize #1}
\def \ETx {\mathcal{E}}
\def \cR {\mathcal{R}}
\def \TRx {\Gamma}	
	
\makeatletter
\let\OldStatex\Statex
\renewcommand{\Statex}[1][3]{%
  \setlength\@tempdima{\algorithmicindent}%
  \OldStatex\hskip\dimexpr#1\@tempdima\relax}
\makeatother
\algdef{SE}[DOWHILE]{Do}{DoWhile}{\algorithmicdo}[1]{\algorithmicwhile\ #1}

\title{Optimal Offline and Competitive Online Strategies for Transmitter-Receiver Energy Harvesting}

\author{Siddhartha Satpathi, Rushil Nagda, Rahul Vaze
\thanks{This paper  will appear in part in Proc. IEEE ICC 2015.}
}

\begin{document}
\maketitle
\thispagestyle{empty}
\pagestyle{empty}
\begin{abstract}
A joint transmitter-receiver energy harvesting model is considered, where both the transmitter and receiver are powered by (renewable) energy harvesting source. Given a fixed number of bits, the problem is to find the optimal transmission power profile at the transmitter and ON-OFF profile at the receiver to minimize the transmission time. With infinite capacity at both the transmitter and receiver, 
optimal offline and optimal online policies are derived. The optimal online policy is shown to be two-competitive in the arbitrary input case.  With finite battery capacities at both ends, only random energy arrival sequence with given distribution are considered, for which an online policy with bounded expected competitive ratio is proposed.
\end{abstract}


\section{Introduction}
Extracting energy from nature to power communication devices has been an emerging area of research.  Starting with \cite{ozel2012achieving, SharmaEH2014}, a lot of work has been reported on finding the capacity, approximate capacity \cite{dong2014near}, structure of optimal policies \cite{sinha2012optimal}, optimal power transmission profile  \cite{UlukusEH2011b, UlukusEH2011c, michelusi2012optimal, VazeEHICASSP14}, competitive online algorithms \cite{VazeEH2011}, etc. One thing that is common to almost all the prior work is the assumption that energy is harvested only at the transmitter while the receiver has some conventional power source. This is clearly a limitation, however, helped to get some critical insights into the problem. 

In this paper, we broaden the horizon, and study the more general problem when energy harvesting is employed both at the transmitter and the receiver. 
The joint (tx-rx) energy harvesting model has not been studied in detail and only some preliminary results are available, e.g., a constant approximation to the maximum throughput has been derived in \cite{VazeEH2014} or \cite{Yener2012enernodes}, \cite{Sharma2010enernodes}.
This problem is fundamentally different than using energy harvesting only at the transmitter, where receiver is always assumed to have energy to stay {\it on}. In contrast to the variable power model at the transmitter where it can choose to transmit any power level given the available energy constraint,  the receiver energy consumption model is binary, as it uses a fixed amount of energy to stay {\it on}, and is {\it off} otherwise. Since useful transmission happens only when the receiver is \textit{on}, the problem is to find jointly optimal decisions about transmit power and receiver ON-OFF schedule. Under this model, there is an issue of coordination between the transmitter and the receiver to implement the joint decisions, however, we ignore that in the interest to make some analytical progress, and assume that the decisions are made by a centralized controller.

We study the canonical problem of finding the optimal transmission power and receiver ON-OFF schedule to minimize the time required for transmitting a fixed number of bits, first in the case when there is no limit on the battery capacities and then generalize it for finite battery capacities at both the transmitter and the receiver. 
We first consider the offline case, where the energy arrivals both at the transmitter and the receiver are assumed to be known non-causally. Even though offline scenario is unrealistic, it still gives some design insights. 
Then we consider the more useful online scenario, where both the transmitter and the receiver only have causal information about the energy arrivals. To characterize the performance of an online algorithm, typically, the metric of competitive ratio is used that is defined as the maximum ratio of `profit' of the online and the offline algorithm over all possible inputs. 


For the infinite battery capacity case, in prior work \cite{UlukusEH2011b}, an optimal offline algorithm has been derived for the case when energy is harvested only at the transmitter, which cannot be generalized with energy harvesting at the receiver together with the transmitter. To understand  the difficulty, assume that the receiver can be \textit{on} for maximum time $T$.
The policy of \cite{UlukusEH2011b} starts transmission at the first energy arrival time, and power transmission profile is the
one that yields the tightest piecewise linear energy consumption curve that lies under the energy harvesting curve at all times and touches the energy harvesting curve at end time.
 The policy of \cite{UlukusEH2011b}, however, may take more than $T$ time and hence may not be feasible with the receiver {\it on} time constraint.
So, we may have to either delay the start of transmission and/or keep stopping in-between  to accumulate more energy to transmit with higher power for shorter bursts, such that the total time for which transmitter and receiver is \textit{on}, is less than $T$. Similarly, for the finite battery capacity, an optimal offline algorithm has been derived for the case when energy is harvested only at the transmitter in \cite{Yener2012optbat}. However, once again there is no easy way of extending the results of \cite{Yener2012optbat}, when both the transmitter and receiver are powered by EH, and we need a new approach.

%
%
%
%

With infinite battery capacity at both the transmitter and the receiver, in the offline scenario, we derive the structure of the optimal algorithm, and then propose an algorithm that is shown to satisfy the optimal structure. The power profile of the proposed algorithm is fundamentally different than the optimal offline algorithm of \cite{UlukusEH2011b}, however, the two algorithms have some common structural properties. 
The recipe of our solution is to first solve the simpler problem of finding the optimal offline algorithm when there is only one energy arrival at the receiver. Building upon this solution, we then derive the optimal offline solution to the problem with multiple energy arrivals at the receiver, to be one among finitely many solutions of the problem with only one energy arrival at the receiver, where corresponding single energy arrivals are suitably constructed. This technique not only gives an elegant method to prove the optimality, but also helps in simplifying the complexity of the optimal algorithm.  

Next, we consider the more useful setup of online algorithms that use only causal information. 
With infinite battery capacities at both ends, for the online scenario, we propose an online algorithm, which starts at time where the accumulated energy at both the transmitter and the receiver is sufficient to transmit the given number of bits eventually. The transmit power at any time (only updated at energy arrival epoch of the transmitter) is such that using the available energy, the remaining number bits are transmitted in minimum time assuming no more energy is going to arrive in future.
We
show that the competitive ratio of the proposed online algorithm is strictly less than $2$ for any energy arrival inputs, even if chosen by an adversary. With only energy harvesting at the transmitter, a $2$-competitive online algorithm has been derived in \cite{VazeEH2011}. This result is more general with different proof technique that allows energy harvesting at the receiver. To prove that the proposed online algorithm is optimal, we show a lower bound on the competetive ratio that is arbitrarily close to $2$ for any online algorithm. This is accomplished by constructing two ``bad" sequences of energy arrivals at the transmitter and the receiver, for which any algorithm fails to achieve a competitive ratio of better than $2$ for at least one of the two sequences.

Finally, we consider the case of finite battery capacity. With finite battery capacity, it is easy to show that the competitive ratio of any online algorithm with the worst case input is unbounded as follows. Suppose, by time slot $t$, any online algorithm consumes more (less) energy than the optimal offline algorithm, then it is easy to construct future energy arrival sequences, for which the optimal offline algorithm can finish transmission of given number of bits, on account of knowing the input sequence and transmitting at a slower (faster) rate, while the online algorithm can never finish the transmission.
Thus, we restrict ourselves to scenario where energy arrivals follow a 
known distribution, but the realization information is only known causally. We propose a simple Accumulate and Dump algorithm, that waits for battery to fill up to a certain prefixed level, and as soon as the accumulated energy is above the level, uses all the energy in the next slot, and restarts accumulating all over again. We show that the expected competitive ratio of the proposed algorithm is finite, which can be computed explicitly given the energy arrival distribution. In prior work \cite{Yener2012optbat,erkal2013optimal, ozel2012optimal}, optimal offline algorithm has been derived when only the transmitter is powered with EH and has a finite battery capacity. Instead of the offline regime, in this paper, we concentrate on the online setting which is more relevant in practice and propose algorithms that have a finite penalty with respect to the optimal offline algorithm.

\vspace{-0.3cm}
\section{System Model}
\vspace{-0.15cm}
\input{NotationsICC}
\section{OPTIMAL OFFLINE ALGORITHM FOR SINGLE ENERGY ARRIVAL AT THE RECEIVER}
\input{OptimalOfflineICC}

\section{OFFLINE ALGORITHM FOR RECEIVER WITH MULTIPLE ENERGY ARRIVALS}
\input{Algo2}

\section{ONLINE ALGORITHM}

\input{onlineICC}

\section{ONLINE ALGORITHMS WITH FINITE BATTERY AT TRANSMITTER AND RECEIVER}
\label{sec:finite_battery}
\input{online_battery}

\section{SIMULATION RESULTS}
\input{experiments}

\section{Conclusions}
In this paper, we have made significant progress in finding optimal transmission strategies when EH is employed at both the transmitter and the receiver. As is evident, EH at both ends is fundamentally different than the case when only the transmitter is powered by EH. With EH at both ends, we have not only found an optimal offline algorithm, which has been accomplished for many other similar but simpler models in past, but also proposed ``good" online algorithms for both finite and infinite battery capacities that have provably efficient competitive ratio compared to the offline algorithms. In particular, in the infinite battery case, the proposed online algorithm is also shown to be optimal. One limitation of tx-rx EH model that we glossed over 
is if there is no centralized controller, how to make transmitter and receiver aware of each others' battery states. This is actually a fundamental issue, and it would require more sophisticated techniques to solve this more general problem. Some limited results are available in  \cite{VazeEH2014}.

\bibliographystyle{IEEEtran}
\bibliography{IEEEabrv,Research}
\appendices

\renewcommand{\theequation}{A.\arabic{equation}}
\numberwithin{equation}{section}
\section{Proof of Lemma \ref{lemma_walds}}
\label{AppendixA}
We seek to apply Wald's equation from Lemma \ref{walds_equation}, for which we have to prove that $\mathbf{E}[N]$ and $\mathbf{E}[\ETx_0]$ are finite.\\
$\mathbf{E}[\ETx_0]<\infty$ follows from the fact that $f(x)=0$ for $x>\mathcal{C}_t$.\\
We now proceed to prove that $\mathbf{E}[N]<\infty$.

\begin{equation}
\mathbf{E}[N]=\sum_{n=1}^{\infty}P[N>n]\le  \sum_{n=1}^{\infty} P\Bigg{[}\displaystyle \sum_{i=0}^{n-1} \ETx_i\le\frac{\mathcal{C}_t}{c}\Bigg{]}.
\label{eq_wald_1}
\end{equation}
Let us choose a constant $x\in (0, \mathcal{C}_t]$ such that $P[\ETx_i>x]>0$ and say $q=P[\ETx_i>x]$.
Define $Y_i = \mathbbm{1}_{\ETx_i>x}$. Clearly all $Y_i$'s are i.i.d random variables. Let $k=\big{\lceil}\frac{\mathcal{C}_t}{cx}\big{\rceil}\Rightarrow kx>\frac{B}{c}$. Now, for any $n>0$ ,
\begin{align}
\displaystyle \sum_{i=0}^{n-1} \ETx_i\le\frac{\mathcal{C}_t}{c}&\Rightarrow \displaystyle \sum_{i=0}^n Y_i\le k,
\\
\nonumber P\Bigg{[}\displaystyle \sum_{i=0}^{n-1} \ETx_i\le\frac{\mathcal{C}_t}{c}\Bigg{]}&\le P\Bigg{[}\displaystyle \sum_{i=0}^{n-1} Y_i\le k\Bigg{]},
\\
&= \displaystyle \sum_{r=0}^k {n\choose r}q^r(1-q)^{n-r}.
\label{eq_wald_2}
\end{align}
From \eqref{eq_wald_1} and \eqref{eq_wald_2},
\begin{align}
&\mathbf{E}[N]\le\displaystyle \sum_{n=1}^\infty \displaystyle \sum_{r=0}^k {n\choose r}q^r(1-q)^{n-r},
\\
& \stackrel{(a)}{\le} \sum_{n=1}^\infty q'^n\displaystyle \sum_{r=0}^k {n\choose r},
\\
& \stackrel{(b)}{\le} \alpha \sum_{n=1}^\infty q'^n n^{k+1}\stackrel{(c)}{<}\infty.
\end{align}
where  $q'=\min(q,1-q)$ in $(a)$. As $\displaystyle \sum_{r=0}^k {n\choose r}$ is a polynomial in $n$ with degree ${k+1}$, $(b)$ follows with some constant $\alpha$. $(c)$ follows since sequence $q'^nn^{k+1}$ converges in $n$, which can be easily verified with the ratio test.
 
Therefore, with $\mathbf{E}[\ETx_0]<\infty$ and $\mathbf{E}[N]< \infty$, we use Wald's equation to write,
\begin{equation}
\mathbf{E}[N]{\mathbf{E}[\ETx_0]}={\mathbf{E}\left[\sum_{i=0}^{N-1} \ETx_i \right]},
\end{equation}
under stopping condition $H$ defined in \eqref{this_is_H}.

\section{}
\label{AppendixB}
From Lemma \ref{lemma_walds}, $\mathbf{E}[\mathcal{N}]$
\begin{align}
&=\nonumber\frac{\mathbf{E}\left[\displaystyle\sum_{i=0}^{\mathcal{N}-1} \ETx_i \right]}{\mathbf{E}[\ETx_0]},
\\
&=\nonumber\frac{\mathbf{E}\left[\mathbf{E}\left[\displaystyle\sum_{i=0}^{\mathcal{N}-1} \ETx_i|H \right]\right]}{\mathbf{E}[\ETx_0]},
\\
&\nonumber\stackrel{(a)}{=}\frac{\mathbf{E}\left[\mathbf{E}\left[\mathbf{E}\left[\displaystyle\sum_{i=0}^{\mathcal{N}-1} \ETx_i\Big{|} \ETx_{\mathcal{N}-1}\ge \mathcal{C}_t/c-k, \displaystyle\sum_{i=0}^{\mathcal{N}-2} \ETx_i=k\right]\right]\right]}{\mathbf{E}[\ETx_0]},
\\
&\nonumber=\frac{\mathbf{E}\left[\mathbf{E}\left[\mathbf{E}\left[k+ \ETx_{\mathcal{N}-1}\Big{|} \ETx_{\mathcal{N}-1}\ge \mathcal{C}_t/c-k, \displaystyle\sum_{i=0}^{\mathcal{N}-2} \ETx_i=k\right]\right]\right]}{\mathbf{E}[\ETx_0]},
\\
&\nonumber=\frac{\mathbf{E}\left[\mathbf{E}\left[\mathbf{E}\left[k+ \ETx_{\mathcal{N}-1}\Big{|} \ETx_{\mathcal{N}-1}\ge \mathcal{C}_t/c-k\right]\right]\right]}{\mathbf{E}[\ETx_0]},
\\
&\nonumber\stackrel{(b)}{\le}\frac{\mathbf{E}\left[\mathbf{E}\left[k+\mathcal{C}_t/c-k+ \mathbf{E}\left[ \ETx_{\mathcal{N}-1}\right]\right]\right]}{\mathbf{E}[\ETx_0]},
\\
&\nonumber=\frac{\mathbf{E}\left[\mathbf{E}\left[\mathcal{C}_t/c+ \mathbf{E}\left[ \ETx_{0}\right]\right]\right]}{\mathbf{E}[\ETx_0]}
\\
&=\frac{\mathcal{C}_t/c}{\mathbf{E}[\ETx_0]}+1,
\label{hitting}
\end{align}
where $k$ is a constant in $(a)$ with $0\le k < \mathcal{C}_t/c$ and $(b)$ follows under Assumption \ref{new_condition}.
For general energy arrival distributions, we can write  $\mathbf{E}[\mathcal{N}]$ as,
\begin{align}
\mathbf{E}[\mathcal{N}]&=\nonumber\frac{\mathbf{E}\left[\displaystyle\sum_{i=0}^{\mathcal{N}-1} \ETx_i \right]}{\mathbf{E}[\ETx_0]},
\\
&=\nonumber\frac{\mathbf{E}\left[\displaystyle\sum_{i=0}^{\mathcal{N}-2} \ETx_i \right]}{\mathbf{E}[\ETx_0]}+\frac{\mathbf{E}\left[ \ETx_{\mathcal{N}-1} \right]}{\mathbf{E}[\ETx_0]},
\\
&\nonumber\stackrel{(a)}{\le}\frac{\mathcal{C}_t/c}{\mathbf{E}[\ETx_0]}+\frac{\mathbf{E}\left[ \ETx_{\mathcal{N}-1} \right]}{\mathbf{E}[\ETx_0]},
\\
&\stackrel{(b)}{\le}\frac{\mathcal{C}_t/c}{\mathbf{E}[\ETx_0]}+\frac{\mathcal{C}_t }{\mathbf{E}[\ETx_0]},
\end{align}
where $(a)$ follows under stopping condition $H$ defined in  \eqref{this_is_H} and $(b)$ follows since, $f(x)=0$ for $x>\mathcal{C}_t$.
 

\end{document}

%% file: NotationsICC.tex
The energy arrival instants at transmitter are marked by $\tau_i$'s with energy $\ETx_i$'s for $i \in \{0,1,\cdots\}$. The total energy harvested at the transmitter till time $t$ is given by \begin{equation}
\ETx(t)=\sum\limits_{i:\tau_i \le t}\ETx_i.
\end{equation} 
Similarly, the energy arrival instants at the receiver are denoted as $r_i$ with energy $\cR_i$. We initialize $\tau_0, r_0$ to $0$ without affecting the system model as follows. If $r_0\le\tau_0$, i.e. the first energy arrival at the receiver occurs before the first energy arrival at the transmitter, then we assume that  $\sum_{i:r_i\le \tau_0}\cR_i$ energy is harvested at the receiver at time $\tau_0$, i.e. $r_0=\tau_0$. We shift the time origin to $\tau_0=r_0$, i.e. $\tau_0=r_0=0$. Note that, since the transmitter has $0$ energy to transmit before time $\tau_0$, no transmission policy can start transmission before $\tau_0$. Therefore, assuming $r_0=\tau_0$ whenever $r_0\le\tau_0$, does not affect any transmission policy. 
Similarly, whenever $\tau_0<r_0$, we assume $\sum_{i:\tau_i\le r_0}\ETx_i$ energy arrives at the transmitter at time $r_0$, i.e. $\tau_0=r_0$, and we offset time origin to $\tau_0=r_0=0$.

The receiver spends a constant $P_{r}$ amount of power to be in `\textit{on}' state during which it can receive data from the transmitter. When it is in `\textit{off}' state it does not receive data, and uses no power. Hence, each energy arrival of $\cR_i$ adds $\Gamma_i= \frac{\cR_i}{P_{r}}$ amount of receiver {\it on} time. The total `time' harvested at the receiver till time $t$ is given by, 
\begin{equation}
\TRx(t)=\sum\limits_{i:r_i \le t}\TRx_i.
\end{equation}

 

The rate of transmission using transmit power $p$ when the receiver is {\it on} is given by a function $g(p)$ which is assumed to follow the following properties,
\begin{align*}
&\text{P1) } &&g(p)\text{ is monotonically increasing in } p, \text{ such that } \ g(0)=0\text{ and }\lim_{p\rightarrow \infty} g(p)= \infty,
\\
&\text{P2) } &&g(p)\text{ is concave in nature with } p,
\\ 
&\text{P3) } &&\frac{g(p)}{p} \text{ is convex, monotonically decreasing with } \  p \text{ and } \lim_{p\rightarrow \infty} \frac{g(p)}{p}= 0.\label{property_decreasing}
\end{align*}
Assuming an AWGN channel, $\log$ function is one such example satisfying all the above properties.


Let a transmission policy change its transmission power at time instants $s_i$'s, i.e. $p_i$ is the transmitter power between time $s_i$ and $s_{i+1}$. The receiver is \textit{on} from time $s_i$ to $s_{i+1}$ whenever $p_i\neq 0$ and is \textit{off} only if $p_i=0$. Thus, succinctly, we say that receiver is \textit{on} at time $t$ to mean that transmit power $p_i\neq 0$ for $t\in[s_i, s_{i+1}]$ and receiver is \textit{on}.
The start and the end time of any policy is denoted by $s_1$ and $s_{N+1}$, respectively. Thus, any policy can be represented as $\{\bm{p}$, $\bm{s}, N\}$, where $\bm{p}=\{p_1, p_2, \cdots , p_N\}$ and $\bm{s}=\{s_1, s_2, \cdots , s_{N+1}\}$. The energy used by a policy at the transmitter upto time $t$ is denoted by $U(t)$, and the number of bits sent by time $t$ is represented by $B(t)$. Clearly, for $j=\argmax_{i}\{s_i < t\}$,
\begin{align}
U(t)&=\;\;\mathclap{\sum_{i=1, p_i\neq 0}^{j-1}}\;\;\; p_i(s_{i+1}-s_i)+p_{j}(t-s_j), \;\; s_1 < t \le s_{N+1},
\\
    &= U(s_{N+1}), \;\; t>s_{N+1},
\\
 &=0, \;\; t\le s_1, 
\\ 
B(t) &=\;\mathclap{\sum_{i=1, p_i\neq 0}^{j-1}}\;\;\; g(p_i)(s_{i+1}-s_i)+g(p_{j})(t-s_j), s_1 < t \le s_{N+1} ,
\\
    &= B(s_{N+1}),\;\; t>s_{N+1},
\\
 &=0, \;\; t\le s_1.
\end{align}
Similarly, the total time for which the receiver is {\it on} till time $t$ is denoted as $C(t)$.

Except for section \ref{sec:finite_battery}, we assume that an infinite battery capacity is available both at the transmitter and the receiver to store the harvested energy. 
Our objective is, given a fixed number of bits $B_0$, minimize the time of their transmission. For any policy, the total time for which the receiver is \textit{on} is referred to as the `transmission time' or the `transmission duration',  and the time by which the transmission of $B_0$ bits is finished, is called as the `finish time'.
Thus, we want to minimize the finish time. Also, since the receiver may not be always \textit{on} before finish time, we have transmission time less than or equal to finish time. Formally, we want to solve,

%

\begin{align}
&\min_{\{\bm{p},\bm{s},N\},T=s_{N+1}}			&& T \label{pb2}
\\
&\text{subject to} 				&& B(T)=B_0, 
\label{pb2_constraint_bits}
\\
&     										&& U(t)\le \mathcal{E}(t)  		\;\;\;\;\;\; \forall \; t\;\in\;[0,T], \label{pb2_constraint_energy}
\\
&    										&&C(t) \leq \TRx(t) \;\;\;\;\;\;\forall \; t \; \in \; [0,T]. 
\label{pb2_constraint_time}
\end{align}
Under transmission policy $\{\bm{p},\bm{s},N\}$, the total receiver \textit{on} time till time $t$ for $s_1<t\le s_{N+1}$ is given by,
\begin{equation}
C(t)=\displaystyle \sum_{i=1}^{k-1} \mathbbm{1}_i(s_{i+1} - s_i) + \mathbbm{1}_k(t - s_{k}),
\label{reciever_used}
\end{equation}
where $k = \max\{i| s_{i} < t\}$ and $\mathbbm{1}_i:\mathbb{R}\rightarrow\{0,1\}$ is a function that takes value $1$ if $p_i>0$ and $0$ if $p_i=0$. 
Constraints \eqref{pb2_constraint_energy} and \eqref{pb2_constraint_time} are the energy neutrality constraints at the transmitter and the receiver, i.e.  energy/on-time used cannot be more than available energy/on-time

%% file: OptimalOfflineICC.tex

In this section, we consider an offline scenario, i.e., all energy arrival epochs $\tau_i$'s and energy harvest amounts $\ETx_i$'s at the transmitter are known ahead of time non-causally.
Moreover,
we assume that the receiver gets only one energy arrival of $\cR_0$ at time $0$, and hence the total receiver {\it on} time is $\Gamma_0= \frac{\cR}{P_r}$.
The crux of problem in both cases (with single/multiple energy arrivals at the receiver) lies in overcoming the problem of the limited transmission time available at the receiver and is not affected much by the number of energy harvests at the receiver. As we shall see, the optimal offline algorithm with multiple energy arrivals at the receiver (solving \eqref{pb2}) consists of repeated application of the derived optimal algorithm for the single energy arrival case. Hence, we postpone the analysis with multiple energy arrivals at the receiver to section \ref{sec:OFFM}.

With only one energy harvest at the receiver, i.e. with  total receiver time $\TRx_0$ harvested at time $0$, a special case of \eqref{pb2}  to minimize the finish time of transmission of $B_0$ bits is,

\begin{align}
\min_{\{\textbf{p},\textbf{s},N\},T=s_{N+1}}	\;\;\;\;\;\;\;\;	&T\label{pb1}
\\
\text{subject to}\;\;\;\;\;\;\; B(T) &=B_0, 
\label{pb1_constraint_bits}
\\
    										 U(t) &\le \mathcal{E}(t), \;\;\; \forall \; t\;\in\;[0,T], \label{pb1_constraint_energy}
\\
    										\sum_{i=1:p_i\neq 0}^{N} (s_{i+1}-s_i)  &\le \TRx_0.
\label{pb1_constraint_time}
\end{align}
Compared to the no receiver constraint \cite{UlukusEH2011b}, Problem \eqref{pb1} is far more complicated, since it involves jointly solving for optimal transmitter power allocation and time for which to keep the receiver {\it on}.

\input{propositionsICC}
\input{Algo1_siddICC}

%% file: propositionsICC.tex
We next present some structural results on the optimal policy to  \eqref{pb1} starting with
Lemma \ref{lemma_increasing_power}, which states that transmission powers in the optimal policy to \eqref{pb1} are non-decreasing over time. 
\begin{lemma}
In an optimal solution to Problem \eqref{pb1}, if $p_i\neq 0$, then $p_i\ge p_j$ $\ \forall \ j<i$ with $i,j\in \{1,2,\cdots ,N\}$\footnote{\label{note1}Observe that without receiver energy harvesting constraint \eqref{pb1_constraint_time}, $p_i\neq 0,\forall i$ from \cite{UlukusEH2011b} and Lemma \ref{lemma_increasing_power} would be same as Lemma 1 in \cite{UlukusEH2011b}. But, as we have constraint on the total receiver time, in optimal solution, transmitter may shut \textit{off} for some time and resume transmission when enough energy is harvested. Hence, $p_i$ may be $0$ in-between transmission. Lemma \ref{lemma_increasing_power} shows that even if this happens, non-zero powers still remain non-decreasing.}.
 
\label{lemma_increasing_power}
\end{lemma}

\begin{proof}
We prove this by contradiction. Assume that the optimal policy (say $X$), with $\{\bm{p},\bm{s},N\}$ violates the condition stated in Lemma \ref{lemma_increasing_power}. Let $p_i\neq 0$ be the first transmission power such that $\exists k<i:\ p_i<p_k $. Let $j=\max \{k:p_i<p_k\}$. 

$Case\;1:$ Suppose $j=i-1$. This situation is shown in Fig. \ref{Lemma1} (a). In this case, consider a new transmission policy (say $Y$) which is same as the optimal policy till time $s_{i-1}$. From $s_{i-1}$ to $s_{i+1}$, $Y$ transmits at a constant power $p'=\dfrac{p_i(s_{i+1}-s_{i})+p_{i-1}(s_{i}-s_{i-1})}{s_{i+1}-s_{i-1}}$. Then the number of bits transmitted by policy $Y$ from time $s_{i-1}$ to $s_{i+1}$ is given by $g(p')(s_{i+1}-s_{i-1})$ while the optimal policy transmits $g(p_i)(s_{i+1}-s_{i})+g(p_{i-1})(s_{i}-s_{i-1})$ bits. Due to concavity of $g(p)$,
\begin{align*}
&g(p_i)\frac{s_{i+1} -s_{i}}{s_{i+1} -s_{i-1}}+g(p_{i-1})\frac{s_{i}-s_{i-1}}{s_{i+1}-s_{i-1}}
\\
& \hspace{20 mm} \le g\left(\frac{p_i(s_{i}-s_{i-1})+p_{i-1}(s_{i+1}-s_{i})}{s_{i+1}-s_{i-1}}\right),
\\
&  g(p_i)(s_{i+1}-s_{i}) +g(p_{i-1})(s_{i}-s_{i-1})
\\
&\hspace{20 mm}\le g(p')(s_{i+1}-s_{i-1}) . 
\end{align*}
Hence, both $X$ and $Y$ transmit equal number of bits till time $s_{i-1}$, while $Y$ transmits more number of bits than $X$ by time $s_{i+1}$. After time $s_{i+1}$, suppose policy $Y$ transmits with power same as policy $X$ till it completes transmitting $B_0$ bits. Since $Y$ has transmitted more bits than $X$ till time $s_{i+1}$, it finishes transmitting all $B_0$ bits earlier than $X$, contradicting the optimality of $X$.

$Case\;2:$ When $j<i-1$, by our assumption on choosing $j$, $p_i>p_{j+1},\cdots ,p_{i-1}$ and $p_i<p_{j}$. So, $p_{i-1},\cdots ,p_{j+1}<p_j$. If any of $p_{i-1},\cdots ,p_{j+1}$ is non zero, then $i$ no longer remains the minimum index violating the condition stated in Lemma \ref{lemma_increasing_power}. Hence, $p_{i-1},\cdots ,p_{j+1}=0$. This situation is shown in Fig. \ref{Lemma1}(b). Now, consider a policy $W$ where the transmission power is same as the optimal policy before time $s_j$ and after time $s_{i+1}$. From $s_j$ to $s_j'=s_j+s_{i}-s_{j+1}$, $W$ keeps the receiver \textit{off} (so transmitter does not transmit in this duration) and from $s_j'$ to $s_{i}$ it transmits at power $p_j$. This policy still transmits equal number of bits and ends at the same time as the optimal policy $X$. Now that $W$ matches with the form of $X$ in \textit{Case 1} from time $s_j'$ to $s_{i+1}$, we could proceed to generate another policy form $W$ (like $Y$ in \textit{Case 1}) which would finish earlier than $W$. Hence, this new policy would finish earlier than $X$ as well and we would reach a contradiction. 

\end{proof}

\begin{figure}[htb]
  \centering
  \centerline{\includegraphics[width=8cm]{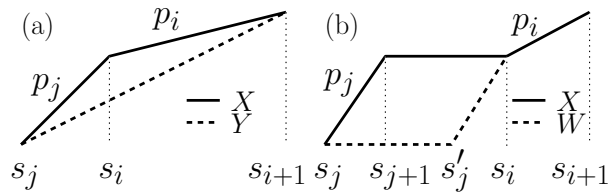}}
\caption{Figure showing the two cases of Lemma \ref{lemma_increasing_power}, (a)\textit{Case 1}   and (b)\textit{Case 2}, with $p_i>p_j$.}\label{Lemma1}
\end{figure}
Although, Lemma \ref{lemma_increasing_power} is valid for every optimal policy to \eqref{pb1}, we will narrow down the search for optimal solutions by looking at an interesting property presented in
Lemma \ref{lemma_nobreaks}, which tells us that there is no need to stop in-between transmissions, and start again. Thus, without affecting optimality, the start of the transmission can be delayed so that transmission power is non-zero throughout. 
\begin{lemma}
The optimal solution to Problem \eqref{pb1} may not be unique, but there always exists an optimal solution where once the transmission has started, the receiver remains `\textit{on}' throughtout, until the transmission is complete. \label{lemma_nobreaks}
\end{lemma}

\begin{proof}
We construct an optimal solution for which $p_i>0$ for all $i\in\{1,\cdots ,N\}$, i.e., with no breaks in transmission, from any other optimal solution. Let an optimal policy $X$ be characterized by $\{\bm{p},\bm{s},N\}$. Now, if $p_i\neq 0\; \forall \ i$, then we are done. Suppose some powers, say $p_{i_1},p_{i_2},\cdots ,p_{i_k}=0$ for some $k<N$, where $i_1<i_2<\cdots <i_k$. We first look at instant $i_1$.

Consider Fig. \ref{fig_Lemma2} (a), and a new policy (say $Y$) which is same as policy $X$ before time $s_{i_1-1}$ and after time $s_{i_1+1}$. But, it keeps the receiver \textit{off} for a duration of $(s_{i_1+1}-s_{i_1})$ starting from time $s_{i_1-1}$ (i.e. from $s_{i_1-1}$ to $s_{i_1}'=(s_{i_1-1}+s_{i_1+1}-s_{i_1})$) and transmits with power $p_{i_1-1}$ from time $s_{i_1}'$ till $s_{i_1+1}$. $Y$ transmits same amount of bits in same time as $X$ and also satisfies constraints \eqref{pb1_constraint_bits}-\eqref{pb1_constraint_time}. So $Y$ is also an optimal policy. But the receiver \textit{off} duration in $Y$, $(s_{{i_1+1}}-s_{i_1})$, has been shifted to left. 

Next, we generate another policy $Z$ from $Y$ by shifting the \textit{off} duration $s_{i_1}'-s_{i_1-1}=(s_{{i_1+1}}-s_{i_1})$ to start from epoch $s_{i_1-2}$ upto $s_{i_1-1}'$, $s_{i_1-1}'-s_{i_1-2}=s_{i_1}'-s_{i_1-1}=(s_{{i_1+1}}-s_{i_1})$, as shown Fig. \ref{fig_Lemma2} (b). $p_{i_1-2}$  is shifted right to start from $s_{i_1-1}'$. Note that $Z$ is also optimal. We continue this process of shifting the receiver \textit{off} period to the left to generate new optimal policies till we reach a policy (say $W$) where the receiver is \textit{off} for time $(s_{{i_1+1}}-s_{i_1})$ from $s_1$, i.e. from $s_{1}$ to $s_1'$, $s_1'-s_1=(s_{{i_1+1}}-s_{i_1})$, as shown in Fig. \ref{fig_Lemma2} (c). As $W$ has $0$ transmission power from the start time $s_1$ to $s_1'$, the effective start time of $W$ can now be changed to $s_1'$. 

We can repeat this procedure for each \textit{off} period corresponding to $p_{i_2},\cdots ,p_{i_k}$ till the total \textit{off} period is shifted to the beginning of transmission. 
This results in a policy with no zero powers in between, that starts \textit{after} time $s_1$ (at $s_1+(s_{{i_1+1}}-s_{i_1})+\cdots +(s_{{i_k+1}}-s_{i_k})$) and ends at the same time $s_{N+1}$ as policy $X$.

\end{proof}
\begin{figure}[htb]
  \centering
  \centerline{\includegraphics[width=8cm]{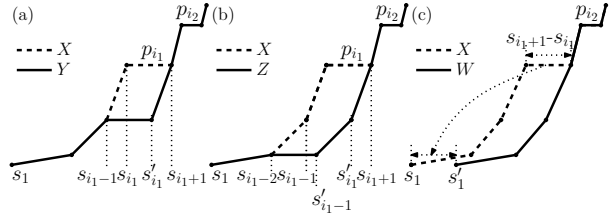}}
\caption{Illustration of Lemma \ref{lemma_nobreaks}. Receiver \textit{off} time of $(s_{j}-s_{i_1})$ is progressively shifted to left as shown in (a) to (b) to (c).}\label{fig_Lemma2}
\end{figure}
\textit{In the subsequent discussion, the optimal solution to Problem \eqref{pb1} means one with no breaks in transmission (reception).} As we shall see in Theorem \ref{th_algo1_1}, such an optimal solution is unique.

Next, we show that the transmission power changes (if at all) only at energy arrival epochs $\tau_i$'s, and the energy used up by that epoch is equal to all the energy that has arrived till then.
\begin{lemma}
For optimal policy $\{\bm{p},\bm{s},N\}$, $s_i=\tau_j$ for some $j$, $U(s_i)=\ETx(s_i^-)\ \forall i\in\{2,\cdots ,N\}$, and $U(s_{N+1})=\ETx(s^-_{N+1})$.
\label{lemma_energy_consumed} 
\end{lemma}
\begin{proof}
By Lemma \ref{lemma_increasing_power} and \ref{lemma_nobreaks},  $p_i\neq 0$ and $p_{i+1}\ge p_i,\forall 1\le i\le N$. So, the proof follows similar to Lemma 2,3 in \cite{UlukusEH2011b}. 
\end{proof}
\vspace{-0.1in}

It may happen that at some epoch $\tau_k$, $U(\tau_k) = \ETx(\tau_k^-)$ holds true, but the transmission power does not change. For notational simplicity, we include all such $\tau_k$'s in $\bm{s}$, where $U(\tau_k)= \ETx(\tau_k^-)$.

Next lemma states that if we take any feasible policy, $\{\bm{p},\bm{s},N\}$ and decrease  $p_1$ and increase $p_N$ while keeping the number of transmitted bits fixed, the transmission time increases, while reducing the finish time of the policy. Lemma \ref{lemma_increase_time} will be useful to prove uniqueness of the optimal policy with no breaks in transmission.
\begin{figure}
\centering
  \centerline{\includegraphics[width=8cm]{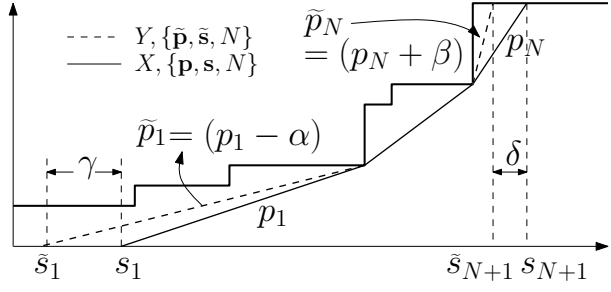}}
\caption{Illustration for the proof of Lemma \ref{lemma_increase_time}.}\label{lemma4}
\end{figure} 
\begin{lemma}
Consider two policies $X$, $\{\bm{p},\bm{s},N\}$  and $Y$, $\{\bm{\widetilde{p}},\bm{\widetilde{s}},N\}$, which are feasible with respect to energy constraint \eqref{pb1_constraint_energy}, have non-decreasing powers and transmit same number of bits in total. If $Y$ is same as $X$ from time $s_2$ to $s_{N}$, but $\widetilde{p}_1=p_1-\alpha,\widetilde{p}_N=p_N+\beta$ with $\alpha,\beta>0$ and $U(s_{N+1})=U(\widetilde{s}_{N+1})$, then we have that the finish time with $Y$ is less than that of $X$, i.e.,  $\widetilde{s}_1=s_1-\gamma, \widetilde{s}_{N+1}=s_{N+1}-\delta$ with some $\gamma,\delta>0$, and the transmission time of $Y$ is more than that of $X$, i.e., $
(\widetilde{s}_{N+1}-\widetilde{s}_1)>(s_{N+1}-s_1)$. 
\label{lemma_increase_time}
\end{lemma}
\begin{proof}
$X$ and $Y$ having used same amount of energy from $s_2$ to $s_{N+1}$, we can say that $\ETx(\widetilde{s}_2)-\ETx(\widetilde{s}_1)=\ETx(s_2)-\ETx(s_1)$, and $\ETx(\widetilde{s}_{N+1})-\ETx(\widetilde{s}_N)=\ETx(s_{N+1})-\ETx(s_N)$. Thus, we can define $\gamma=\dfrac{\alpha}{p_1-\alpha}(s_2-s_1)$ and $\delta =\dfrac{\beta}{p_N+\beta}(s_{N+1}-s_N)$. As $X$ and $Y$ transmit equal number of bits in total and are identical between time $s_2$ and $s_N$, we can just equate the number of bits transmitted by $X$ before $s_1$ and after $s_{N}$ (LHS of \eqref{Lemma4_eq1}) with that of $Y$ (RHS of \eqref{Lemma4_eq1}), i.e.,
\begin{align}
&g(p_N)(s_{N+1}-s_N)+g(p_1)(s_2-s_1)\nonumber
\\
&=g(\widetilde{p}_N)(\widetilde{s}_{N+1}-\widetilde{s}_N)+g(\widetilde{p}_1)(\widetilde{s}_2-\widetilde{s}_1)\label{Lemma4_eq1},
\end{align}
(Note that only one of the four variable $\alpha,\beta,\gamma,\delta$ can be independently chosen.)
Therefore, from \eqref{Lemma4_eq1},
\begin{align}
& (p_N+\beta)p_N\frac{\delta}{\beta}\left(\frac{g(p_N)}{p_N}-\frac{g(p_N+\beta)}{p_N+\beta}\right)\nonumber
\\
&=(p_1-\alpha)p_1\frac{\gamma}{\alpha}\left(\frac{g(p_1-\alpha)}{p_1-\alpha}-\frac{g(p_1)}{p_1}\right).\label{bits_equal}
\end{align}
As $g(p)/p$ is a continuous \& differentiable function, the mean value theorem implies that $\exists$ $p_N':p_N<p_N'<p_{N}+\beta$ and $p_1':p_1-\alpha<p_1'<p_{1}$ such that
\begin{align}
&\frac{d}{dp} \frac{g(p)}{p} \bigg{\vert}_{p=p_N'}=\frac{1}{\beta}\left(\frac{g(p_N+\beta)}{p_N+\beta}-\frac{g(p_N)}{p_N}\right)\text{ and }\label{diff_1}
\\
&\frac{d}{dp} \frac{g(p)}{p}\bigg{\vert}_{p=p_1'}=-\frac{1}{\alpha}\left(\frac{g(p_1-\alpha)}{p_1-\alpha}-\frac{g(p_1)}{p_1}\right)\label{diff_2}.
\end{align}
Substituting \eqref{diff_1} and \eqref{diff_2} in \eqref{bits_equal} we get,
\begin{align}
&\delta p_Np_N\frac{d}{dp} \frac{g(p)}{p}  \bigg{\vert}_{p=p_N'}
=\gamma p_1'p_1\frac{d}{dp} \frac{g(p)}{p} \bigg{\vert}_{p=p_1'}.\label{bits_equal1}
\end{align}
Now $\dfrac{d}{dp} \dfrac{g(p)}{p}$ is an increasing function of $p$ since $g(p)/p$ is convex. Hence, with $p_1' <p_1\le p_N<p_N'$, \begin{equation}
\frac{d}{dp} \frac{g(p)}{p}  \bigg{\vert}_{p=p_N'}>\frac{d}{dp} \frac{g(p)}{p} \bigg{\vert}_{p=p_1'}.
\end{equation}Thus, (\ref{bits_equal1}) implies $\gamma >\delta$. So, transmission time in the policy $Y$, $\left( s_{N+1}-s_1+\gamma-\delta\right)$, is greater than the transmission time in policy $X$ i.e. $(s_{N+1}-s_1)$.
\end{proof}

Lemma \ref{transmission_duration} uses Lemma \ref{lemma_increase_time} to prove that if the start time of the optimal policy is delayed beyond the first `time' arrival instant $r_0=0$ at the receiver, then the transmission time will be equal to  $\TRx_0$, i.e., it will exhaust all the transmission time available with the receiver.
\begin{lemma}  For an optimal policy $\{\bm{p},\bm{s},N\}$, either $s_{N+1} - s_1 = \TRx_0$ or $s_1 = r_0= 0$. 
\label{transmission_duration}
\end{lemma}
\begin{proof} We use contradiction to prove the result.
Suppose the optimal policy say $X$, starts at $s_1>0$ and has transmission time $(s_{N+1}-s_1)< \TRx_0$. 
We will generate another policy which has finish time less than that of $X$, having transmission time squeezed in between $(s_{N+1}-s_1)$ and $\TRx_0$.
Consider policy $Y$ ($\{\bm{\widetilde{p}},\bm{\widetilde{s}},N\}$) in relation to $X$, as defined in Lemma \ref{lemma_increase_time}. As $\alpha$, $\beta$, $\delta$, $\gamma$ are all related (by constraints presented in Lemma \ref{lemma_increase_time}), choice of one variable (we consider $\alpha$) defines $Y$. By definition of $s_i$'s, $s_{2}$ is the first energy arrival which is on the boundary of energy constraint (\ref{pb1_constraint_energy}) i.e. $U(s_2)=\ETx(s_2^-)$ and $s_{N}$ is the last epoch satisfying $U(s_N)=\ETx(s_N^-)$. Hence, we can choose $\alpha>0$, such that $\widetilde{p}_1$ and $\widetilde{p}_N$ would be feasible with respect to energy constraint (\ref{pb1_constraint_energy}). Note that if $s_1=0$, then any value of $\alpha$ would have made $\widetilde{p}_1$ infeasible. 

From Lemma \ref{lemma_increase_time}, we know that the transmission time of policy $Y$ is more than that of $X$, i.e. $(\widetilde{s}_{N+1}-\widetilde{s}_1) > (s_{N+1} - s_{1})$. From the hypothesis $(s_{N+1}-s_1) < \TRx_0$. Therefore, let $(s_{N+1}-s_1)=\TRx_0-\epsilon$, with $\epsilon >0$. If the chosen value of $\alpha$ is such that $\gamma -\delta\le\epsilon$, then $\left( \widetilde{s}_{N+1}-\widetilde{s}_1\right)<\TRx_0$. If not, then we can further reduce $\alpha$ so that $\gamma -\delta\le\epsilon$ ($\alpha$,$\beta$,$\gamma$,$\delta$ being related by continuous functions).  Note that, when $\epsilon=0$, any choice of $\alpha$ would make $\left( \widetilde{s}_{N+1}-\widetilde{s}_1\right)>\TRx_0$. Hence, with this choice of $\alpha$, $(s_{N+1}-s_1)<\left( \widetilde{s}_{N+1}-\widetilde{s}_1\right)<\TRx_0$ holds and  policy $Y$ 
contradicts the optimality of policy $X$ (as finish time of $Y$ is less than finish time of $X$, $\widetilde{s}_{N+1}=s_{N+1}-\delta <s_{N+1}$ from Lemma \ref{lemma_increase_time}). Thus $s_{N+1}-s_1=\TRx_0$ if $s_1\neq 0$ in an optimal policy.
\end{proof}

%% file: Algo1_siddICC.tex

Summarising the results of Lemmas \ref{lemma_increasing_power}-\ref{transmission_duration}, the optimal policy $\{\bm{p},\bm{s},N\}$ may change transmission powers only at energy arrival epochs i.e. $\forall\;i\in \{2,\cdots , N\},\ s_i=\tau_j$ for some $j$. At these epochs, it exhausts the total energy available i.e. $U(s_i)=\ETx(s_i^-)$. The transmission powers are also non-decreasing with time, and the optimal policy uses up the total `receiver time' allowed, if it does not start transmitting from $r_0=0$.

	Now we prove in Theorem \ref{th_algo1_1} that the structure described in Lemma \ref{lemma_increasing_power}-\ref{transmission_duration} including Lemma \ref{lemma_Q} (for ease of presentation Lemma \ref{lemma_Q} is postponed to section \ref{sec:OFF}) is not only necessary, but is indeed sufficient for optimality of a policy. 
\begin{theorem}
A policy $\{\bm{p},\bm{s},N\}$ is an optimal solution to Problem \eqref{pb1} if and only if, 
\label{th_algo1_1}
\begin{align}
&\sum_{i=1}^{i=N}g(p_i)(s_{i+1}-s_i)=B_0; 							
\label{claim1}
\\
&p_1\le p_2 \cdots \le p_N;
\label{claim3}  
\\
&\nonumber s_i=\tau_j  \ \ \ \ \ \ \ \ \ \ \ \ \ \ \ \text{ for some } j, i\in \{2,\cdots ,N\} \ \text{ and }
\\
& U(s_i)=\ETx(s_i^-), \ \ \ \ \ \forall i\in \{2,\cdots ,N+1\};
\label{claim4}
\\
&\nonumber s_{N+1}-s_1=\TRx_0,  \ \ \ 						\text{ if } s_1>0 \text{ or }
\\
& s_{N+1}\le \TRx_0,			\ \ \ \ \ \ \ \ \				\text{ if } s_1=0;
\label{claim2}
\\
&\exists s_j:s_j\in \bm{s} \text{ and } s_j=\tau_q,
\label{claim5}
\end{align}
where $\tau_q$ is defined in INIT\_POLICY of section \ref{sec:OFF}.
\end{theorem}
\begin{proof}

The proof consists of establishing both necessary and sufficiency conditions. The necessity of  (\ref{claim1}) follows as it is a constraint to the Problem \eqref{pb1},  \eqref{claim3} follows from Lemma \ref{lemma_increasing_power}, \ref{lemma_nobreaks}, \eqref{claim4} follows from Lemma \ref{lemma_energy_consumed}, \eqref{claim2} follows from Lemma \ref{transmission_duration}, and \eqref{claim5} follows from Lemma \ref{lemma_Q}.

Now, we prove the sufficiency of the structure \eqref{claim1}-\eqref{claim5}. Let a policy $X$, $\{\bm{p},\bm{s},N\}$ follow  structure \eqref{claim1}-\eqref{claim5}. We need to show that this policy is optimal, which we do via contradiction. Suppose $X$ is not optimal. Let there exists another policy $Y$, $\{\bm{p'},\bm{s'},N'\}$ which is optimal. Since $Y$ abides by Lemma \ref{lemma_increasing_power}-\ref{lemma_Q} on account of its optimality, $Y$ also satisfies structure \eqref{claim1}-\eqref{claim5}. (Now both $X$ and $Y$ satisfy structure \eqref{claim1}-\eqref{claim5} but $Y$ is optimal i.e. it finishes before $X$. This would would mean that there \textit{possibly} exists some more conditions which are followed by $Y$ but not $X$). We need to show that such a optimal policy $Y$ (different from $X$) cannot exist or is infeasible, i.e., both $X$ and $Y$ cannot simultaneously  satisfy \eqref{claim1}-\eqref{claim5} and be different.\footnote{Note that Lemma \ref{lemma_nobreaks} suggests that optimal solution to Problem \eqref{pb1} may not be unique in general, but Theorem \ref{th_algo1_1} shows that the optimal solution \textit{without breaks} in transmission is indeed unique.} 

The following cases arise depending on whether $s_1'>s_1$, $s_1'=s_1$ or $s_1'<s_1$.




\textit{Case1}: If $s_1'>s_1\ge 0$, then by \eqref{claim2}, $s_{N'+1}'=s_1'+\TRx_0>s_1+\TRx_0\ge s_{N+1}$. So policy $Y$ finishes after time $s_{N+1}$ and hence cannot be optimal. 

\textit{Case2}: Suppose $s_1'=s_1$. Let $s_i'$ be the first epoch for which $p_i'\ne p_i$ for some $i \in \{1,2,\cdots ,N\}$. 


Suppose $p_i'>p_i$. If, in policy $Y$, transmission continues after $s_{i+1}$ i.e. $s_{N'+1}'>s_{i+1}$, then the amount of energy used by $Y$ in interval $[s_{i},s_{i+1}]$ can be lower bounded by $p_i'(s_{i+1}-s_i)$, which follows from \eqref{claim3}. Since $p_i'>p_i$, $p_i'(s_{i+1}-s_i)$ is more than $p_i(s_{i+1}-s_i)$, which is the energy used by policy $X$. But by structure \eqref{claim4}, $X$ uses all energy available at both $s_i$ and $s_{i+1}$. So, the maximum energy available in $[s_{i},s_{i+1}]$ is $p_i(s_{i+1}-s_i)$. Therefore, $Y$ uses more than available energy in $[s_{i},s_{i+1}]$ and is not feasible with respect to the energy constraint. 

If $s_{N'+1}'\le s_{i+1}$, then it can be easily verified by concavity of function $g(p)$ that $Y$ transmits strictly less number of bits in interval $[s_i,s_{N'+1}]$ than $X$ in interval $[s_{i},s_{i+1}]$. Both policies being same till $s_i$, we conclude that $Y$ transmits less than $B_0$ bits by its finish time $s_{N'+1}$, and thus it is not feasible with respect to \eqref{claim1}.

When $p_i>p_i'$, symmetrical arguments follow.

\textit{Case3}: This case argues the infeasibility of $Y$ when $0\le s_1'<s_1$. Since $s_1>0$, transmission time of $X$ is equal to $\TRx_0$ from \eqref{claim2}. The idea of the proof is to show that if an optimal policy $Y$ starts its transmission early and finishes earlier than policy $X$, it always takes more transmission time than $X$ ($=\TRx_0 $), which is going to violate the time constraint (\ref{pb1_constraint_time}). First, we establish that $Y$ must be same as policy $X$ from epoch $s_2$ to an epoch $s_j$ such that $s_j=\displaystyle\max_{s_i<s_{N'+1}'} s_i$. Let $s_k'=\displaystyle\max_{s_i'<s_2}s_i'$, and $Y$ continue from $s_k'$ with constant power $p_k'$ till $s_{k+1}'$. Clearly $s_{k+1}'\ge s_2$ from definition of $s_{k}'$. 

Suppose $s_{k+1}'>s_2$. Since transmission with a constant power $p_k'$ from $s_k'$ to $s_{k+1}'$ is feasible,  transmission with constant power $\dfrac{\ETx(s_2^-)-\ETx(s_{k}'^-)}{(s_2-s_{k}')}$ from $s_{k}'$ to $s_2$, and $\dfrac{\ETx(s_{k+1}'^-)-\ETx(s_{2}^-)}{(s_{k+1}'-s_2)}$ from $s_2$ to $s_{k+1}'$ is also feasible for any policy (Refer to Fig. \ref{Theorem1_figure} (a)) and hence,
\begin{equation}
\dfrac{\ETx(s_{k+1}'^-)-\ETx(s_{2}^-)}{(s_{k+1}'-s_2)}<\dfrac{\ETx(s_2^-)-\ETx(s_{k}'^-)}{(s_2-s_{k}')}.
\label{th1_eq1}
\end{equation} 
Transmission with power $\dfrac{\ETx(s_2^-)-\ETx(s_{k}'^-)}{(s_2-s_{k}')}$ exhausts all available energy at epochs $s_{k}'$ and $s_2$. Therefore, power $p_1=\dfrac{\ETx(s_2^-)}{(s_2-s_{1})}$ (in policy $X$) from $s_1$ to $s_2$ must be greater than $\dfrac{\ETx(s_2^-)-\ETx(s_{k}'^-)}{(s_2-s_{k}')}$. If not, then transmission with power $p_1$ in $X$ would become infeasible. Thus, from \eqref{th1_eq1}, \begin{equation} 
\dfrac{\ETx(s_{k+1}'^-)-\ETx(s_{2}^-)}{(s_{k+1}'-s_2)}<p_1.
\label{th1_eq2}
\end{equation}

Now, transmission with power $p_1$ from $s_1$ to $s_2$, and transmission with power $\dfrac{\ETx(s_{k+1}'^-)-\ETx(s_{2}^-)}{(s_{k+1}'-s_2)}$ from $s_2$ to $s_{k+1}'$ are both feasible for any policy. This combined with \eqref{th1_eq2} would imply transmission with a constant power $\dfrac{\ETx(s_{k+1}'^-)}{(s_{k+1}'-s_1)} $ from $s_1$ to $s_{k+1}'$ is feasible and hence,
\begin{equation}
\dfrac{\ETx(s_{k+1}'^-)}{(s_{k+1}'-s_1)}<p_1.
\label{th1_eq3}
\end{equation}

%


Since finish time of $X$, $s_{N+1}=s_1+\TRx_0>s_1'+\TRx_0\ge s'_{N'+1}\ge s'_{k+1}$, $X$ transmits in interval $[s_1,s'_{k+1}]$ and uses atleast $p_1(s'_{k+1}-s_1)$ energy in this interval, which follows from \eqref{claim3}. But, the maximum energy available for transmission  in interval $[s_1,s'_{k+1}]$ is $\ETx(s_{k+1}'^-)$. From	 \eqref{th1_eq3}, we can infer that $X$ uses more than this available energy in $[s_1,s'_{k+1}]$, and therefore, we reach a contradiction over feasibility of $X$. So, our hypothesis, $s_{k+1}'>s_2$, is incorrect. Since, $s_{k+1}'\ge s_2$, we can conclude that $s_{k+1}'=s_2$. 

Now, let $p_{k+1}'\neq p_2$ and $s_j>s_3$. From the definition of $p_2$, $p_{k+1}>p_2$. Then the amount of energy used by policy $Y$ between $s_2$ and $s_3$ is more than what is available. So $p_{k+1}'=p_2$ ($s'_{k+2}=s_3$) and similarly, we can show that $p'_{k+2}=p_3\cdots$ ($ s'_{k+3}=s_4\cdots$) till epoch $s_j$. This completes the proof that $Y$ is same as policy $X$ from epoch $s_2$ to $s_j$.

By structure (\ref{claim5}) we can be sure that there exists atleast one epoch $s_i=\tau_q$ which belongs to $\bm{s}$ as well as $\bm{s'}$. So, $j\ge 2$. 

\begin{figure}[htb]
\centering
\centerline{\includegraphics[width=8cm]{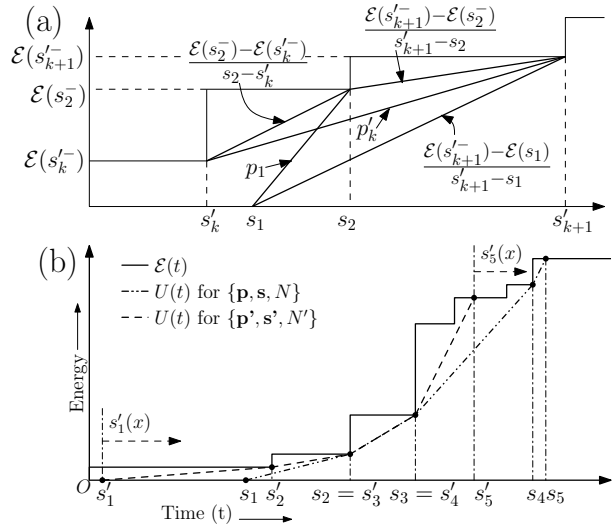}}
\caption{Energy curves at transmitter explaining \textit{Case3} in proof of Theorem \ref{th_algo1_1}}
\label{Theorem1_figure}
\end{figure}

Continuing with \textit{Case3}, consider the following process which creates feasible policies from policy $\{\bm{p'},\bm{s'},N'\}$ as shown in Fig. \ref{Theorem1_figure} (b). We define two pivots $l$ and $r$. Initially we set $l=s_2'$ and $r=s_{N'}'$. The transmission power right before $l$ is $u$ ($u=p_1'$ initially) and right after $r$ is $v$ ($v=p_{N'}'$ initially). Keeping the policy $\{\bm{p'},\bm{s'},N'\}$ same from $l$ to $r$, we increase $u$ by a small amount to $u+du$ and decrease $v$ by a small amount to $v-dv$ such that the number of bits transmitted (i.e. $B_0$) remains same under this transformation. This would lead to change in the start time $s_1'$ and finish time $s_{N'+1}'$. Let the starting time of transmission $s_1'$ change to $s_1'+x$ and the finish time $s_{N'+1}'$ change to $s_{N'+1}'+y$ for some $x,y>0$ (note that $y$ is dependent on $x$). We denote such a policy by vectors $\{\bm{p'(x)},\bm{s'(x)},N'(x)\}$. 

Following Lemma \ref{lemma_increase_time}, we can conclude that $(s_{N'(x)+1}'(x)-s_1'(x))<(s_{N'+1}'-s_1')$. We continue increasing $x$ till either $u=p_2'(x)$ (in which case we change $l=s_3'(x)$) or $v=p_{N'-1}'(x)$ (where we change $r=s_{N'-1}'(x)$) or $s_{N'(x)+1}'(x)$ hits an epoch, say $\tau_j$ (we change $r=\tau_j$, $v\rightarrow\infty$ in this case). After this, we again start increasing $x$ with changed definitions of $l,r,u,v$. 
We continue this process till $x=s_1-s_1'$  or $u$ becomes equal to $v$. 
Note that the value of $x$ for which $u$ becomes equal to $v$, would be greater than $(s_1-s_1')$, since policy $\{\bm{p'(x)},\bm{s'(x)},N'(x)\}$ shares at least one epoch with policy $X$, by arguments of previous paragraph. By maintaining these rules on $l,r,u,v$, we ensure that policy $\{\bm{p'(x)},\bm{s'(x)},N'(x)\}$ abides by structure \eqref{claim1}-\eqref{claim4}, \eqref{claim5} and is feasible with energy constraint. Since $\left( s_{N'(x)+1}'(x)-s_1'(x)\right)$ is decreasing with $x$, and $\left( s_{N'(0)+1}'(0)-s_1'(0)\right)=s_{N+1}'-s_1\le \TRx_0$, the policy $\{\bm{p'(x)},\bm{s'(x)},N'(x)\}$ is also feasible with receiver  time constraint. At $x=s_1-s_1'$, we reach a policy such that $s_1'(x)=s_1$. For $x=s_1-s_1'$, if $s_{N'(x)+1}'(x)\ge s_{N+1}$ then $s_{N'+1}'-s_1'>s_{N'(x)+1}'(x)-s_1'(x)\ge s_{N+1}-s_1=\TRx_0$ and policy $Y$ is infeasible with receiver  time constraint. If $s_{N'(x)+1}'(x)< s_{N+1}$, then we can follow the arguments presented in \textit{Case2} to show that policy $\{\bm{p'(x)},\bm{s'(x)},N'(x)\}$ (at $x=s_1-s_1'$) is infeasible, which in turn shows the infeasibility of policy $Y$.
\end{proof}

\section{Optimal Offline Algorithm }
\label{sec:OFF}
In this section, we propose an offline algorithm $\mathsf{OFF}$ for Problem \eqref{pb1}, and show that it satisfies the sufficiency conditions of Theorem \ref{th_algo1_1}.
Algorithm $\mathsf{OFF}$ first finds an initial feasible solution via INIT\_POLICY, and then iteratively improves upon it via PULL\_BACK. Finally, QUIT produces the output.

\begin{algorithm}
\caption{OFF}
\label{Algorithm1}
\begin{algorithmic}[1]
\State \textbf{Input}: $\ETx(t),B_0,\TRx_0$.
	\State $\{\bm{p},\bm{s},N\}\gets$ INIT\_POLICY($\ETx(t)$,$B_0$,$\TRx_0$).
	\State $X\gets$ PULL\_BACK($\{\bm{p},\bm{s},N\}$).
	\State $\{\bm{p},\bm{s},N\}\gets$ QUIT($X$).
	\State \Return 	$\{\bm{p},\bm{s},N\}$. 
\end{algorithmic}
\end{algorithm}

\subsection{INIT\_POLICY} 
\textit{Idea: } Initially, we find a constant power policy that is feasible and starts as early as possible. Also, we try to make it satisfy most of the sufficiency conditions of Theorem \ref{th_algo1_1}.

\textit{Step1:} Identify the first energy arrival instant $\tau_n$, so that using $\ETx(\tau_n)$ energy and $\TRx_0$ time, $B_0$ or more bits can be transmitted with a constant power (say $p_c$), i.e. $\TRx_0g\left(\dfrac{\ETx(\tau_n)}{\TRx_0}\right)\ge B_0$. Then solve for $\widetilde{\TRx}_0$,
\begin{small}
\begin{equation}
\widetilde{\TRx}_0\, g\left(\dfrac{\ETx(\tau_n)}{\widetilde{\TRx}_0}\right)= B_0,\ p_c = \dfrac{\ETx({\tau_n})}{\widetilde{\TRx}_0}.
\label{INIT_POLICY_time}
\end{equation}
\end{small}
\begin{figure}
\centering
  \centerline{\includegraphics[width=8cm]{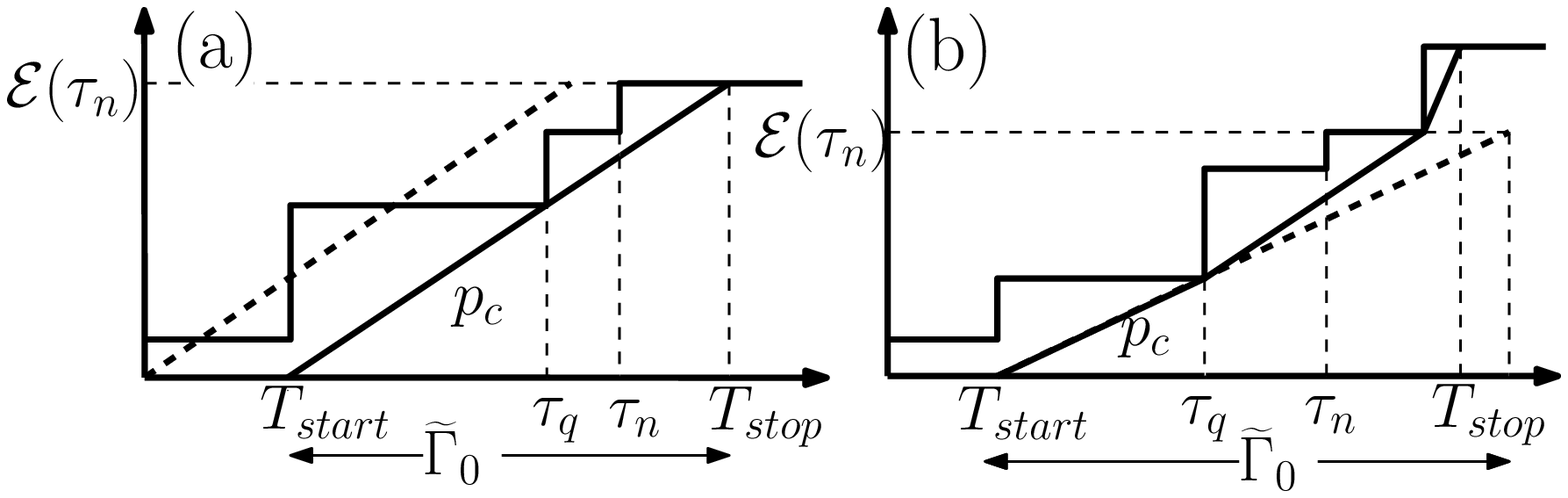}}
\caption{Figure showing point $\tau_q$.}\label{straight}
\end{figure} 
\textit{Step2:} Find the earliest time $T_{start}$, such that transmission with power $p_c$ from $T_{start}$ for $\widetilde{\TRx}_0$ time is feasible with energy constraint \eqref{pb1_constraint_energy}. Set $T_{stop} = T_{start} + \widetilde{\TRx}_0$. Let $\tau_q$ be the \textit{first epoch}, where $U(\tau_q) = \ETx(\tau_q^-)$ (Fig. \ref{straight}).
Lemma \ref{lemma_Q} shows that point $\tau_q$ thus found leads to a `good' initial solution as, in every optimal solution total harvested energy till $\tau_q$ is used up at $\tau_q$. This in-turn implies that $\tau_q\in \bm{s}$, if $\{\bm{p},\bm{s},N\}$ is the optimal policy.

If $U(T_{stop}) = \ETx(T_{stop}^-)$ as shown in Fig. \ref{straight}(a), then terminate INIT\_POLICY with constant power policy $p_c$. 

Otherwise, if $U(T_{stop}) < \ETx(T_{stop}^-)$, then 
modify the transmission after $\tau_q$ as follows. Set $\widetilde{B}_0 = (T_{stop} - \tau_q)g(p_c)$,  which denotes the number of bits left to be sent after time $\tau_q$. Then apply Algorithm 1 of \cite{UlukusEH2011b} from time $\tau_q$ to transmit $\widetilde{B}_0$ bits in as minimum time as possible without considering the receiver {\it on} time constraint. 
Update $T_{stop}$, to where this policy ends. So, $U(T_{stop}) = \ETx(T_{stop}^-)$ from \cite{UlukusEH2011b}. Since Algorithm 1 \cite{UlukusEH2011b} is optimal, it takes minimum time ($=T_{stop}-\tau_q$) to transmit $\widetilde{B}_0$ starting at time $\tau_q$. 
However, using power $p_c$ to transmit $\widetilde{B}_0$ takes $(T_{start}+\widetilde{\TRx}_0 - \tau_q)$ time.
Hence, $T_{stop}\le (T_{start}+\widetilde{\TRx}_0)$. 
As $\widetilde{\TRx}_0\le \TRx_0$ from \eqref{INIT_POLICY_time}, $(T_{stop}- T_{start})\le \TRx_0$. This shows that solution thus found using Algorithm 1 \cite{UlukusEH2011b}, is indeed feasible with receiver time constraint \eqref{pb1_constraint_time}. Now, output of INIT\_POLICY is a policy that transmits at power $p_c$ from $T_{start}$ to $\tau_q$, and after $\tau_q$ uses Algorithm 1 of \cite{UlukusEH2011b}.

\begin{figure}
\centering
  \centerline{\includegraphics[width=8cm]{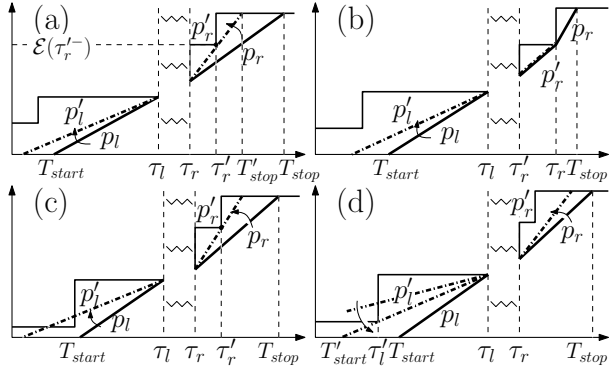}}
\caption{Figures showing possible configurations in any iteration of the PULL\_BACK. The solid line represents the transmission policy in the previous iteration and dash dotted lines are for the current iteration.}\label{figure_Algorithm1}
\end{figure}

\begin{lemma}
In every optimal solution, at energy arrival epoch $\tau_q$ defined in INIT\_POLICY, $U(\tau_q)=\ETx(\tau_q^-)$.
\label{lemma_Q}
\end{lemma}
\begin{proof}
We shall prove this by contradiction. For simplicity of notation let $R=T_{start}$ and $S=T_{stop}$ with $T_{start}$, $T_{stop}$ being the start and finish time of constant power policy $p_c$ defined in INIT\_POLICY. First, we make the following claims:

\textbf{Claim 1:} Every optimal transmission policy begins transmission at or before time $R$.

Since, $S-R=\widetilde{\TRx}_0\le \TRx_0$, by Lemma \ref{transmission_duration}, if a transmission policy has to finish before $S$, it has to start before time $\max(S-\TRx_0,0) \le \max(R,0)=R$. 


\textbf{Claim 2:} Every optimal transmission policy ends transmission at or before time $S$.

If it does not, then constant power policy $p_c$ finishing at $S$ will contradict its optimality.

Suppose we have an optimal transmission policy, say $X$,$\{\bm{p},\bm{s},N\}$, that does not exhaust all its energy at time $\tau_q$ i.e. $U(\tau_q)<\ETx(\tau_q^-)$. Then, by Lemma \ref{lemma_energy_consumed}, it does not change its transmission power at $\tau_q$. Let the transmission power of $X$ be $p_{j-1}$ at $\tau_q$ and $p_{j-1}$ starts from $s_{j-1}$ and goes till $s_j$. Now, $s_j<S$ by \textit{Claim 2}. Further, power $p_c$ exhausts all energy by $\tau_q$. So,
\begin{align}
&p_c(\tau_q-R)=\ETx(\tau_q^-)\label{eqlemmaQ1}.
\end{align}
But, by constraint (\ref{pb1_constraint_energy}),
\begin{align}
&p_c(\tau_q-R)+p_c(s_j-\tau_q)\le \ETx(s_j^-),
\\
& p_c(s_j-\tau_q)\stackrel{(a)}{\le} \ETx(s_j^-)-\ETx(\tau_q^-),
\\
& p_c(s_j-\tau_q)< \ETx(s_j^-)-U(\tau_q)=p_{j-1}(s_j-\tau_q),
\\
& p_c<p_{j-1},\label{eqlemmaQ2}
\end{align}
where $(a)$ follows from (\ref{eqlemmaQ1}).
If ${j-1}= 1$, then power at $\tau_q$ is the first transmission power $p_1$. But then by \eqref{eqlemmaQ2}, $p_1 > p_c$. By the definition of $p_c$, we must have $s_{1} > R$, and this will contradict \textit{Claim 1}.

So ${j-1}\ge 2$, which means that the power of transmission must change at least once between $R$ and $\tau_q$. By Lemma \ref{lemma_energy_consumed}, $X$ has used all energy by $s_{j-1}$ and $s_{j}$. So, $p_{j}(\ETx(s_{j}^-)-\ETx(s_{j-1}^-))$ is the maximum energy available between time $s_{j-1}$ and $s_{j}$. If $R<s_{j-1}$, then $p_c$ (by \eqref{eqlemmaQ2}) uses more energy, than available between $s_{j-1}$ and $s_{j}$, which is not possible. If $s_{j-1}\le R$ then $p_{j-1}$ uses more than maximum energy available (given by $p_c(\tau_q-R)=\ETx(\tau_q^-)$ ) between time $R$ and $\tau_q$, violating energy constraint \eqref{pb1_constraint_energy}. 

Therefore, every optimal transmission policy must use all energy till epoch $\tau_q$. 

\end{proof}

\begin{algorithm}
\caption{INIT\_POLICY}

\footnotesize
\label{init_policy}
\begin{algorithmic}[1]
\State \textbf{Input}:$\ETx(t)$, $B_0$, $\TRx_0$

\State $n=\displaystyle \argmin_k\left(\left\{\tau_k | \TRx_0 g\left(\frac{\ETx(\tau_k)}{\TRx_0}\right)\geq B_0\right\}\right)$. \label{init_policy_Etn}

\State Solve for $\widetilde{T}: \widetilde{T}g\left(\dfrac{\ETx(\tau_n)}{\widetilde{T}}\right) = B_0.$\label{init_policy_CP_time}

\State $p_c=\dfrac{\ETx(\tau_n)}{\widetilde{T}}.$


\State $q=\displaystyle \argmin_{k\in [n]}\ ( \{ \tau_k | IS\_FEASIBLE(\{p_c,p_c\},\{\tau_k-\ETx(\tau_k^-)/p_c,\tau_k,\tau_k+(\ETx(\tau_n)-\ETx(\tau_k^-))/p_c\},2)==1\}).$

\label{init_policy_t_q}
\State $T_{start}=\tau_q-\dfrac{\ETx(\tau_q^-)}{p_c}$, $T_{stop}=\tau_q+\dfrac{\ETx(\tau_n)-\ETx(\tau_q^-)}{p_c}.$

\If {$U(T_{stop})<\ETx(\tau_n)$}
	\State $\widetilde{B}=g(p_c)(T_{stop}-\tau_q).$\label{init_policy_bits_t_q}  
	\State $\{\bm{p},\bm{s},N\}\gets$  Apply Algorithm 1 in \cite{UlukusEH2011b} to 	minimize transmission
		\Statex  time of $\widetilde{B}$ bits  after time $\tau_q$ assuming	a  total of $\ETx_{q}$   
		\Statex amount of energy available at $\tau_q$. 
	\State\Return $\{\{p_c,\bm{p}\},\{T_{start},\bm{s}\},N+1\}$. \label{init_policy_Yang}
\Else 
	\State\Return $\{\{p_c,p_c\},\{T_{start},\tau_q,T_{stop}\},2\}$. \label{init_policy_CP}
\EndIf
\Statex\hrulefill
\Statex $IS\_FEASIBLE({\bm p},{\bm s},N)$ returns $1$ if policy $\{{\bm p},{\bm s},N\}$ is feasible and $0$ otherwise. 
\end{algorithmic}
\end{algorithm}

Now that we have an initial feasible solution, we improve upon this policy iteratively as presented in PULL\_BACK. But, before getting into the formal algorithm, we explain the procedure PULL\_BACK through an example.

\textit{Example PULL\_BACK:}
Assume that the starting feasible solution is given by the constant power policy, as shown by dotted line in Fig. \ref{figure_example_Algorithm1} (a), where $\tau_q=\tau_2$. We first assign the following initial values for the initial feasible policy - transmission power left of $\tau_2$ as $p_l=p_c$, power right of $\tau_2$ as $p_r=p_c$, start time $T_{start}$ stop time $T_{stop}$ as start and stop time of  constant policy power $p_c$, epoch at which $p_l$ ends as $\tau_l=\tau_2$, epoch at which $p_r$ starts as $\tau_r=\tau_2$. Now, we increase $p_r$, till it reaches $p_r'$ which hits the boundary of energy feasibility at epoch $\tau_3$, as shown by the solid line in Fig \ref{figure_example_Algorithm1} (a). Since, in total we need to transmit $B_0$ bits, the decrease in bits transferred by $p_r$ to $p_r'$ (RHS of \eqref{eq_example1}) is compensated by calculating appropriate $p_l'$ according to the following equation, where LHS represents the increase in bits transmitted from $p_l$ to $p_l'$.
\begin{align}
\nonumber g(p_l') & \frac{\ETx(\tau_l^-)}{p_l'}  -g(p_l)(\tau_l-T_{start})
\\
&=-g(p_r)(T_{stop}-\tau_r)
+g(p_r')\frac{\ETx(T_{stop}'^-)-\ETx(\tau_3^-)}{T_{stop}'-\tau_3}.
\label{eq_example1}
\end{align}   
Having got a feasible $p_l'$, as shown in Fig. \ref{figure_example_Algorithm1} (a), we assign $T_{start}'$ with the time at which transmission with power $p_l'$ starts, $T_{stop}'$ with time at which transmission with power $p_r'$ finishes. $\tau_r'$ gets the value $\tau_3$ and $\tau_l'$ remains same as $\tau_l=\tau_2$. Note that parameters $\{T_{start}',T_{stop}',\tau_l',\tau_r',p_l',p_r'\}$ define the policy at the end of first iteration. 

In the next iteration, the portion of transmission between $\tau_l'=\tau_2$ to $\tau_r'=\tau_3$ is not updated. In this iteration, we try to increase $p_r'$ till it hits the feasibility equation \eqref{pb1_constraint_energy} of energy. $p_r'$ could virtually be increased to infinity. But transmission with infinite power for 0 time does not transmit any bits. So we assign $\tau_r''=\tau_2$ and $p_r''=\frac{\ETx(\tau_3^-)-\ETx(\tau_2^-)}{\tau_3-\tau_2}$. With this change of $p_r'$ to $p_r''$, we again calculate $p_l''$ which compensates the decrease in bits transferred after $\tau_r'$. But the calculated $p_l''$ becomes infeasible at $\tau_1$ as shown in Fig. \ref{figure_example_Algorithm1} (b). Hence, we set $p_l''$ to the minimum feasible power $\frac{\ETx(\tau_2^-)-\ETx(\tau_1^-)}{\tau_2-\tau_1}$ as shown in Fig. \ref{figure_example_Algorithm1} (c). With this $p_l''$, we re-calculate $p_r''$, so as to transmit $B_0$ bits in total. $\tau_l''$ is assigned to $\tau_1$, $\tau_r''$ remains $\tau_3$. $T_{start}''$ and $T_{stop}''$ are updated to values marked in Fig. \ref{figure_example_Algorithm1} (c). The final policy at the end of second iteration is shown by solid line in Fig. \ref{figure_example_Algorithm1} (c). Similarly, we continue to the third iteration, by improving the policy at the end of second iteration to finish earlier.

\begin{figure}
\centering
  \centerline{\includegraphics[width=8cm]{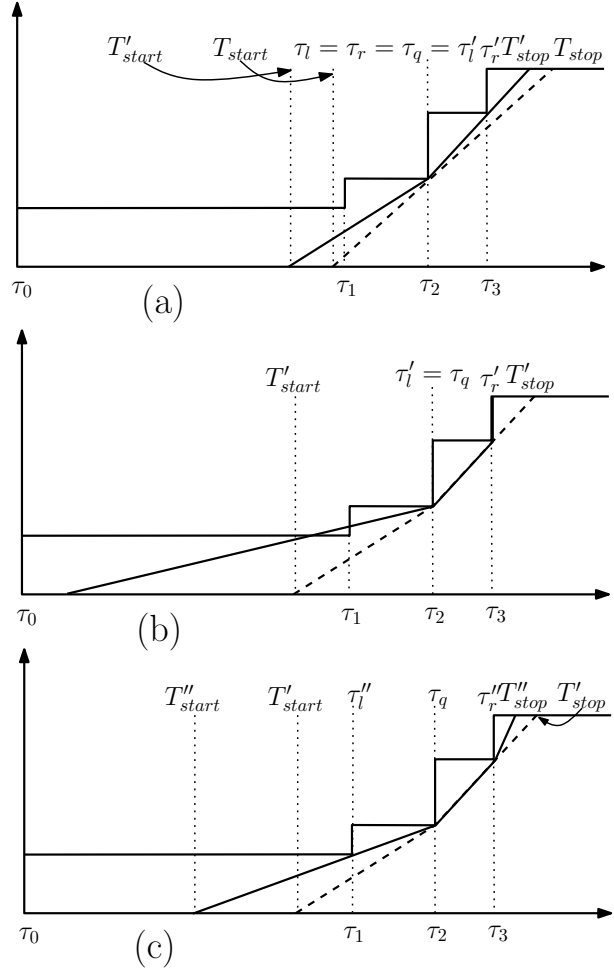}}
\caption{Figures showing (a) first  and (c) second iteration of the PULL\_BACK through an example. (b) representes an intermidiate step in second iteration. In this diagram, the dashed line represent previous iteration policy and solid line is the present iteration policy.}\label{figure_example_Algorithm1}
\end{figure}

\subsection{ PULL\_BACK}
Now, we describe the iterative subroutine PULL\_BACK whose input is policy $\{\bm{p},\bm{s},N\}$ output by INIT\_POLICY.
\\
\textit{Idea: }Clearly, $\{\bm{p},\bm{s},N\}$, the output of INIT\_POLICY, satisfies all but structure \eqref{claim2} of Theorem \ref{th_algo1_1}, since we cannot guarantee whether $(s_{N+1}-s_1)=\TRx_0$ when $s_1>0$. 
So, the main idea of PULL\_BACK is to increase the transmission duration from $(s_{N+1}-s_1)\le \widetilde{\TRx}_0$, in INIT\_POLICY, to $\TRx_0$ in order to satisfy \eqref{claim2}, while decreasing the finish time for reaching the optimal solution. To achieve this, we utilize the structure presented in Lemma \ref{lemma_increase_time} and iteratively increase the last transmission power $p_N$, and decease the first transmission power $p_1$. 
%

Initialize $\tau_l=s_2,\tau_r=s_N,p_l=p_1,p_r=p_N,T_{start}=s_1,T_{stop}=s_{N+1}$. In any iteration, $\tau_{l}$ and $\tau_{r}$ are assigned to the first and last energy arrival epochs, where $U(\tau_l)=\ETx(\tau_l^-)$ and $U(\tau_r)=\ETx(\tau_r^-)$. $p_l$ and $p_r$ are the constant transmission powers before $\tau_l$ and after $\tau_r$, respectively. 
We reuse the notation $\tau$ here, because $\tau_l$ and $\tau_r$ will occur at energy arrival epochs from Lemma \ref{lemma_energy_consumed}. 
$T_{start}$ and $T_{stop}$ are the start and finish time of the policy, found in any iteration. $\tau_l, \tau_r, p_l, p_r, T_{start}, T_{stop}$ get updated to $\tau_l', \tau_r', p_l', p_r', T'_{start}, T'_{stop}$ over an iteration. In any iteration, only one of $\tau_l$ or $\tau_r$ gets updated, i.e., either $\tau_l'=\tau_l$ or $\tau_r'=\tau_r$. Further, PULL\_BACK ensures that \textit{transmission powers between $\tau_l$ and $\tau_r$ do not get changed} over an iteration.
Fig. \ref{figure_Algorithm1} shows the possible updates in an iteration of PULL\_BACK.

\textit{Step1, Updation of $\tau_r$, $p_r$:} Initialize $p_r'=p_r$ and increase $p_r'$ till it hits the boundary of energy constraint \eqref{pb1_constraint_energy}, say at $(\tau_r',\ETx(\tau_r'^-))$ as shown in Fig. \ref{figure_Algorithm1}(a). The last epoch where $p_r'$ hits \eqref{pb1_constraint_energy} is set to $\tau_r'$. So, $U(\tau_r') = \ETx(\tau_r'^-)$. Set $T_{stop}'$ to where power $p_r'$ ends.
Calculate $p_l'$ such that decrease in bits transmitted due to change from $p_r$ to $p_r'$ is compensated by increasing $p_l$ to $p_l'$, via
\begin{align}
\nonumber g(p_r)(T_{stop}&-\tau_r)-g(p_r')(T_{stop}'-\tau_r')
\\
&=g(p_l')\frac{\ETx(\tau_l'^-)}{p_l'}-g(p_l)(\tau_l-T_{start}).
\label{eq_example1}
\end{align}
Suppose, $p_r$ can be increased till infinity without violating \eqref{pb1_constraint_energy}, as shown in Fig. \ref{figure_Algorithm1}(b).  This happens when there in no energy arrival between $\tau_r$ and $T_{stop}$. In this case, set $p_r'$ to the transmission power at $\tau_r^-$. Set $\tau_r'$ as the epoch where $p_r'$ starts, and $T_{stop}'$ to $\tau_r$. Calculate $p_l'$ similar to \eqref{eq_example1}.

\textit{Step2, Updation of $\tau_l, p_l$}: If $p_l'$ obtained from \textit{Step1} is feasible, as shown in Fig. \ref{figure_Algorithm1}(a), set $T_{start}'=\tau_l-\frac{\ETx(\tau_l'^-)}{p_l'}$, $\tau_l'=\tau_l$. Proceed to \textit{Step3}. Otherwise, if $p_l'$ is not feasible, as shown in Fig. \ref{figure_Algorithm1}(c), the changes made to $\tau_r',p_r'$ in \textit{Step1} are discarded. As shown in Fig. \ref{figure_Algorithm1} (d), $p_l'$ is increased from its value in \textit{Step1} until it becomes feasible. $\tau_l'$ is set to the first epoch where $U(\tau_l') = \ETx(\tau_l'^-)$. Similar to \textit{Step1}, calculate $p_r'$ such that the increase in bits transmitted due to change of $p_l$ to $p_l'$ is compensated, and update $T_{stop}'$ accordingly. Set $\tau_r'=\tau_r$. Proceed to \textit{Step3}.

\textit{Step3, Termination condition:} If $T_{stop}' - T_{start}' \ge \TRx_0$ or $T_{start}' = 0$, then terminate PULL\_BACK. Otherwise, update $\tau_l, \tau_r, p_l, p_r, T_{start}, T_{stop}$ to $\tau_l', \tau_r', p_l', p_r', T'_{start}, T'_{stop}$ receptively and GOTO \textit{Step1}.

By design of PULL\_BACK, we know that the finish time decreases at every iteration. Next, in Lemma  \ref{lemma_PULL_BACK_power},  we show that transmission time increases with each iteration of  PULL\_BACK. This is used in Lemma \ref{th_running_time} to establish a bound on the running time of PULL\_BACK.

\begin{lemma}
Transmission time $(T_{stop}-T_{start})$ monotonically increases over each iteration of PULL\_BACK.
\label{lemma_PULL_BACK_power}
\end{lemma}
\begin{proof}
In any iteration of PULL\_BACK, the possible valid configurations can be one of the three shown in Fig. \ref{figure_Algorithm1} (a), (b), (d). Since it is too verbose to describe the three possible cases, we refer to Fig. \ref{figure_Algorithm1}. Note that $\ETx(T_{stop}^-)=\ETx(T_{stop}'^-)$ in (a), (d). In case (b), we can assume that $T_{stop}'=\tau_r^+$ and transmission continues beyond $\tau_r$, but with infinite power. Since transmitting with infinite power for $0$ time does not transmit any bits, we would transmit the same number of bits, as we did prior to this modification. So, $\ETx(T_{stop}^-)=\ETx(T_{stop}'^-)$ in (d) as well. Thus, the transmission policy for two consecutive iterations satisfy the conditions of Lemma \ref{lemma_increase_time} (with $\beta\rightarrow \infty$ for case (d)) and therefore, $(T_{stop}-T_{start})$ increases across constitutive iterations of PULL\_BACK.
\end{proof}
\begin{lemma} Worst case running time of PULL\_BACK is linear with respect to the number of energy harvests before finish time of INIT\_POLICY.
\label{th_running_time}
\end{lemma}
\begin{proof}
%
%
%
%
%
Since, in an iteration of PULL\_BACK, either $\tau_r$ or $\tau_l$ updates, the number of iterations is bounded by the values attained by $\tau_l$, in addition to that of $\tau_r$. Initially, $\tau_l\le \tau_q$ and $\tau_r \ge \tau_q$.
As $\tau_l$ is non-increasing across iterations, $\tau_l\le \tau_q$ throughout.
Assume that $\tau_r$ remains greater than or equal to $ \tau_q$ across INIT\_POLICY. Then, both $\tau_l$ and $\tau_r$ can at max attain all $\tau_i$'s less than finish time of initial feasible policy. Hence, we are done.

It remains to show that $\tau_r\ge \tau_q$. $\tau_n$ is defined as the first energy arrival epoch with which $B_0$ or more bits can be transmitted in $\TRx_0$ time and $\tau_q\le \tau_n$, by definition. 
So, when $T_{stop}$ becomes $\le \tau_n \, or \,  \tau_q$, then transmission time, $(T_{stop}-T_{start})$, should be  $>\TRx_0$. 
But, in the initial iteration $(T_{stop}-T_{start})\le \TRx_0$ and $(T_{stop}-T_{start})$ increases monotonically, from Lemma \ref{lemma_PULL_BACK_power}. Hence, PULL\_BACK will terminate before $T_{stop}$ (and therefore $\tau_r$) decreases beyond $\tau_q$.   
\end{proof}

\begin{algorithm}
\caption{PULL\_BACK}
\footnotesize
\label{Algorithm1}
\begin{algorithmic}[1]
\State \textbf{Input}: $\{\bm{p},\bm{s},N\}\gets$ INIT\_POLICY
\State \textbf{Initialization}: $\tau_l=s_2$, $\tau_{r}=s_{N}$, $T_{start}=s_1$, $T_{stop}=s_{N+1}$, $p_l=p_1$, $p_r=p_{N}$, $control=0$.
\State Delete $\bm{s}.first$, Delete $\bm{s}.last$, Delete $\bm{p}.first$, Delete $\bm{p}.last$.
\While {$\left(T_{stop}-T_{start}< \TRx_0\right) \text{ and } \left(T_{start}>0\right)$}
\label{algo_while_loop}
	\State $\{\tau_l',\tau_r',T_{start}',T_{stop}',p_l',p_r',\bm{p'},\bm{s'}\}$\hspace{8cm}$
=\{\tau_l,\tau_r,T_{start},T_{stop},p_l,p_r,\bm{p},\bm{s}\}.$  
	\If {$\{i:\tau_r<\tau_i < T_{stop}\}=\emptyset$}
		\State $B_l=g(p_r)(T_{stop}-\tau_r)+g(p_l)(\tau_l-T_{start})$, $control=1$.
	\Else
		\State $j=\displaystyle \argmin_{i:\tau_r<\tau_i < T_{stop}} \mathcal{P}(\tau_r,\tau_i)$.
		\State $B_l=g(p_r)(T_{stop}-\tau_r)+g(p_l)(\tau_l-T_{start})-g\left(\mathcal{P}(\tau_r,\tau_j)\right)						\left(\frac{\ETx(T_{stop}^-)-\ETx(\tau_r^-)}{\mathcal{P}(\tau_r,\tau_j)}\right) $.\label{algo_bits_left_1}
	\EndIf

	\State Solve for $\tilde{p}$ in $0<\tilde{p}<p_l$: $\frac{\ETx(\tau_l^-)}{\tilde{p}}g(\tilde{p})=B_l$.\label{algo_bits_left_2}
	\If {$\tilde{p}$ exists}
		\State   $istrue=IS\_FEASIBLE\left(\tilde{p},\{\tau_l-					\ETx(\tau_l^-)/\tilde{p},\tau_l\},1\right).$
		\Else    $\;\;\; istrue=0.$ 
		\EndIf 
		

	\If {$istrue==1$}
		\State $p_l=\tilde{p}$, $T_{start}=\tau_l-\ETx(\tau_l^-)/p_l$.
		\If {$control=0$}
			\State $p_r=\mathcal{P}(\tau_r,\tau_j)$, $T_{stop}=\tau_r+\frac{\ETx(T_{stop}^-)-\ETx(\tau_r^-)}							{\mathcal{P}					(\tau_r,\tau_j)}$.
			\State $\tau_r=\tau_j$, $\bm{p}.append(\mathcal{P}(\tau_r,\tau_j))$, $\bm{s}.append( \tau_j)$.
		\Else
			\State $T_{stop}=\tau_r$.
			\If {$\bm{p}\neq \emptyset$} 
				\State Delete $\bm{s}.last$, $\tau_r=\bm{s}.last$, $p_r=\bm{p}.last$,  Delete $\bm{p}.last$.
			\EndIf
			\State $control=0$.
			
		\EndIf
	\Else 
		\State $k=\displaystyle \argmax_{i:\max\left(\left(\tau_l-\ETx(\tau_l^-)/\tilde{p}\right),\tau_0\right)					\le \tau_i <\tau_l} \mathcal{P}(\tau_i,\tau_l)$.
		\State $B_r=g(p_r)(T_{stop}-\tau_r)+g(p_l)(\tau_l-T_{start})-g\left(\mathcal{P}(\tau_k,\tau_l)\right).							\left(\frac{\ETx(\tau_l^-)}{\mathcal{P}(\tau_k,\tau_l)}\right)$
		\State $p_l=\mathcal{P}(\tau_k,\tau_l)$, $T_{start}=\tau_l-\ETx(\tau_l^-)/\mathcal{P}(\tau_k,\tau_l)$.
		 
		\State $\tau_l=\tau_k$, $\bm{p}.prepend=p_l$, $\bm{s}.prepend=\tau_k$.
	
		\State Solve for $p_r$: $\frac{\ETx(T_{stop}^-)-\ETx(\tau_r^-)}{p_r}g(p_r)=B_r$.
		\State $T_{stop}=\tau_r+(\ETx(T_{stop}^-)-\ETx(\tau_r^-))/p_r$	.
	\EndIf

\EndWhile
\State \Return $\{\tau_l',\tau_r',T_{start}',T_{stop}',p_l',p_r',\bm{p'},\bm{s'},T_{start},T_{stop}\}.$ 

\Statex\hrulefill
\Statex $\mathcal{P}(\tau_i,\tau_j)=\frac{\ETx(\tau_j^-)-\ETx(\tau_i^-)}{\tau_j-\tau_i}.$
\end{algorithmic}

\end{algorithm}

\begin{algorithm}
\caption{QUIT}
\footnotesize
\label{Algorithm1}
\begin{algorithmic}[1]
\State \textbf{Input}:\\$\{\tau_l',\tau_r',T_{start}',T_{stop}',p_l',p_r',\bm{p'},\bm{s'}, T_{start},T_{stop}\}\gets $PULL\_BACK.
\If {$(T_{start}-T_{stop})>\TRx_0$}
	\State $T=\TRx_0-(\tau_r'-\tau_l')$, $B=B_0-\displaystyle\sum_{i}g(p_i')(s_{i+1}'-s_{i}')$.
	\State Solve for  $x$:
		\begin{align*}
		&xg\left(\dfrac{\ETx(\tau_{l}'^-)}{x}\right)+\left(T-x\right) g 		\left(\dfrac{\ETx(T_{stop}'^-)-\ETx(\tau_{r}'^-)}{T-x}\right)=B.
		\end{align*}			
	\label{algo_solve_eqn}
	\State $p_l'=\dfrac{\ETx(\tau_{l}'^-)}{x}$, $T_{start}'=\tau_l'-x$.
	\State $p_r'=\dfrac{\ETx(T_{stop}'^-)-\ETx(\tau_{r}'^-)}{T-x}$, $T_{stop}'=\tau_r'+T-x$.
	\State $\bm{p'}.prepend(p_l')$, $\bm{s'}.prepend(T_{start}')$, $\bm{p'}.append(p_r')$, $\bm{s'}.append(T_{stop}')$.
\State \Return $\{\bm{p'},\bm{s'}, \text{number of elements in } \bm{p'}\}$.
\Else
	\State $\bm{p}.prepend(p_l)$, $\bm{s}.prepend(T_{start})$, $\bm{p}.append(p_r)$, $\bm{s}.append(T_{stop})$.
\State \Return $\{\bm{p},\bm{s}, \text{number of elements in } \bm{p}\}$.  
\EndIf

\end{algorithmic}
\end{algorithm}

The third and final	subroutine of $\mathsf{OFF}$ is QUIT.

\subsection{QUIT}If $T_{start}' = 0$ and $T_{stop}' - T_{start}' \le \TRx_0$ upon PULL\_BACK's termination, then PULL\_BACK's policy at termination is output. Note that structure \eqref{claim2} holds for this policy. Otherwise, if $T_{stop}' - T_{start}' > \TRx_0$ (which happens for the first time across iterations of PULL\_BACK), then we know that in penultimate step $T_{stop} - T_{start} < \TRx_0$.
Hence, we are looking for a policy that starts in $ [T_{start} ,\ T_{start}']$ and ends in $[T_{stop} ,\ T_{stop}']$, whose transmission time is equal to $\TRx_0$. 
We solve for $x,y$ (let the solution be $x^*,y^*$),
\begin{align}
\nonumber (\tau_l-x)& \; g\left(\frac{\ETx(\tau_l^-)}{\tau_l-x}\right)+(y-\tau_r)\; g\left(\frac{\ETx(T_{stop}^-)}{y-\tau_r}\right)\\
&=g(p_l)(\tau_l-T_{start})+g(p_r)(T_{stop}-\tau_r),
\label{eq_termination_0}
\\
y-x&=\TRx_0.
\label{eq_termination}
\end{align}
At penultimate iteration, $(x,y)=(T_{start},T_{stop})$, \eqref{eq_termination_0} is satisfied and $y-x<\TRx_0$.
At $(x,y)=(T_{start}',T_{stop}')$, as $\ETx(T_{stop}^-)=\ETx(T_{stop}'^-)$, \eqref{eq_termination_0} is satisfied and $y-x>\TRx_0$. 
So, there must exist a solution $(x^*,y^*)$ to \eqref{eq_termination_0}, where $x^*\in [T_{start}',T_{start}]$, $y^*\in [T_{stop}',T_{stop}]$ and $y^*-x^*=\TRx_0$, for which, \eqref{claim2} holds. Output with this policy which starts at $x^*$ and ends at $y^*$.

Now, we state Theorem \ref{th_algo1_2} which proves the optimality of Algorithm $\mathsf{OFF}$.
\begin{theorem}
The transmission policy proposed by Algorithm $\mathsf{OFF}$ is an optimal solution to Problem \eqref{pb1}.
\label{th_algo1_2}
\end{theorem}
\begin{proof}
We show that  Algorithm $\mathsf{OFF}$ satisfies the sufficiency conditions of  Theorem \ref{th_algo1_1}.
To begin with, we prove that the power allocations satisfy \eqref{claim3}, by induction. 
First we establish the base case that INIT\_POLICY's output satisfies \eqref{claim3}.
If INIT\_POLICY returns the constant power policy $p_c$ from time $T_{start}$ to $T_{stop}$, then clearly the claim holds. 

Otherwise, INIT\_POLICY applies Algorithm 1 from \cite{UlukusEH2011b} with $\widetilde{B}=B_0-g(p_c)(\tau_q-T_{start})$ bits to transmit after time $\tau_q$. 
Algorithm 1 from \cite{UlukusEH2011b} ensures that transmission powers are non-decreasing after $\tau_q$. 
So we only need to prove that the transmission power $p_c$ between time $T_{start}$ and $\tau_q$ is less than or equal to the transmission power just after $\tau_q$ (say $p_q$), via contradiction. Assume that $p_q<p_c$. Let transmission with $p_q$ end at an epoch $\tau_{q'}$, where $U(\tau_{q'})=\ETx(\tau_{q'}^-)$ form \cite{UlukusEH2011b}.  The energy consumed between time $\tau_q$ to $\tau_{q'}$ with power $p_c$ is,
\begin{equation} 
p_c(\tau_{q'}-\tau_q)>p_q(\tau_{q'}-\tau_q)\stackrel{(a)}{=}\ETx(\tau_{q'}^-)-\ETx(\tau_q^-),\label{eq_1_algo1_modified}
\end{equation} 
where $(a)$ follows from $U(\tau_q)=\ETx(\tau_q^-)$. Further, the maximum amount of energy available for transmission between $\tau_q$ and $\tau_{q'}$ is $\left(\ETx(\tau_{q'}^-)-\ETx(\tau_q^-)\right)$. By \eqref{eq_1_algo1_modified}, transmission with $p_c$ uses more than this energy and therefore it is infeasible between time $\tau_q$ and $\tau_{q'}$. But, by definition of $p_c$, transmission with power $p_c$ is feasible till time $(T_{start}+\widetilde{\TRx}_0)$. Also, $\tau_{q'}\le T_{stop}$ by definition of $\tau_{q'}$ and $T_{stop}\le (T_{start}+\widetilde{\TRx}_0)$. So, power $p_c$ must be feasible till $\tau_{q'}$ and we reach a contradiction.        

Now, we assume that the transmission powers output from PULL\_BACK are non-decreasing till its $n^{th}$ iteration. Therefore, as transmission powers between $\tau_l$ and $\tau_r$ does not change over an iteration, powers would remain non-decreasing in the $(n+1)^{th}$ iteration if we show that $p_l'<p_l$ and $p_r'>p_r$. 
In any iteration, by  definition, either $\tau_l$ or $\tau_{r}$ updates. Assume $\tau_l$ gets updated to $\tau_{l}'$, $p_l$ to $p_l'$, $p_r$ to $p_r'$ and $\tau_r$ remains same, shown Fig. \ref{figure_Algorithm1}(d) (when $\tau_r$ updates, the proof follows similarly).  
Then we are certain that $p_{r}'>p_r$ by algorithmic steps. So from $n^{th}$ to $(n+1)^{th}$ iteration, the number of bits transmitted after $\tau_r$ should decrease. Thus, the number of bits transmitted before $\tau_l$ must be increasing. This implies $p_l'\le p_l$ and this completes the proof for transmission powers being non-decreasing at the end of every iteration of PULL\_BACK.  

Next, we show that QUIT outputs a policy with non-decreasing transmission powers. Let the policy being output by QUIT be $X,\{{\bm p},{\bm s},N\}$. Let the start and finish time of the policy  at the penultimate iteration of PULL\_BACK (say $Y$) be $T_{start}$ and $T_{stop}$, respectively. From the algorithmic design of PULL\_BACK, we know that $Y$ is identical to $X$ from time $s_2$ to $s_{N}$. Also, since $Y$ is a policy from PULL\_BACK, it has non-decreasing transmission powers. Thus, we can write the power profile of $Y$ as $\{\frac{\ETx(s_2'^-)}{s_2-T_{start}}, p_2,p_3, \cdots , p_{N-1}, \frac{\ETx(s_{N+1}'^-)}{T_{stop}-s_N}\}$, where
\begin{equation}
\frac{\ETx(s_2'^-)}{s_2-T_{start}}\le p_2\le\cdots\le p_{N-1}\le \frac{\ETx(s_{N+1}'^-)}{T_{stop}-s_N}. 
\label{eq:quit_monotone_power}
\end{equation}

Hence, in order to prove monotonicity of ${\bm p}$, we only need to show $p_1\le p_2$ and $p_{N-1}\le p_N$. From QUIT, recall that $s_1=x^*\le T_{start}$ and $s_{N+1}=y^*\le T_{stop}$. Thus, $p_1=\frac{\ETx(s_2'^-)}{s_2-s_1}\le \frac{\ETx(s_2'^-)}{s_2-T_{start}}\stackrel{(a)}{\le} p_2$ and $p_{N}=\frac{\ETx(s_{N+1}'^-)}{s_{N+1}-s_N}\ge \frac{\ETx(s_{N+1}'^-)}{T_{stop}-s_N}\stackrel{(a)}{\ge} p_{N-1}$, where $(a)$ follows from \eqref{eq:quit_monotone_power}. 

Hence, transmission powers output by $\mathsf{OFF}$ are non-deceasing and satisfy \eqref{claim3}. Since $\mathsf{OFF}$ transmits equal number of bits (=$B_0$) throughout INIT\_POLICY, PULL\_BACK and QUIT, it satisfies \eqref{claim1}. Clearly, \eqref{claim4} is maintained throughout $\mathsf{OFF}$,  and by arguments presented at end of QUIT, we know that $\mathsf{OFF}$ satisfies \eqref{claim2}.
%



Now consider structure \eqref{claim5}. As $\tau_q$ is present in INIT\_POLICY, the only way 
$\tau_q$ cannot be part of the policy (say $\{{\bm p},{\bm s},N\}$) in an iteration of PULL\_BACK, i.e. $\tau_q\notin {\bf s}$, is when $\tau_r$ decreases beyond $\tau_q$. But $\tau_r\ge \tau_q$ as shown in Lemma \ref{th_running_time}.  So, the policy output by $\mathsf{OFF}$ includes $\tau_q$.
To conclude, $\mathsf{OFF}$ satisfies \eqref{claim1}-\eqref{claim5}, and hence is an optimal algorithm.
\end{proof}

{\it Discussion:} In this section, we solved the special case of \eqref{pb2}, when there is only energy arrival at the receiver. Even this special case is hard, compared to having receiver powered by conventional energy source. We proposed a three phase iterative algorithm, where in first we come up with a reasonable feasible solution and then improve upon it in the next two phases until it satisfies the sufficient conditions for the optimal solution. We use this solution of the special case as a building block to solve the general problem \eqref{pb2} in next section.

%% file: Algo2.tex
\label{sec:OFFM}
We now consider solving the general problem \eqref{pb2} in the offline setting, when receiver harvests energy multiple times.
Our approach to solve problem \eqref{pb2} is to use the algorithm  $\mathsf{OFF}$ repeatedly. Corresponding to every receiver `time' arrival of $\TRx_i$ at $r_i$, let $O_i$ be the earliest time instant such that the receiver can be kept \textit{on} continuously, without any break, from time $O_i$ to $O_i+\TRx(r_i)$ (see Fig. \ref{figure_origin_points} (a)). It can be easily seen that the receiver will exhaust all its available energy (or  attain the boundary of \eqref{pb2_constraint_time}) at atleast one receiver `time' arrival epoch when it is kept \textit{on} from $O_i$ to $O_i+\TRx(r_i)$. If not, then we can start the receiver slightly earlier than $O_i$ and keep it \textit{on} for $\TRx(r_i)$ time without violating constraint \eqref{pb2_constraint_time}, which is contradictory to our definition of $O_i$. For example, in Fig. \ref{figure_origin_points} (a), when the receiver turns \textit{on} from $O_1$, it exhausts all it's energy at $r_1^-$.

\begin{figure}
\centering
  \centerline{\includegraphics[width=8cm]{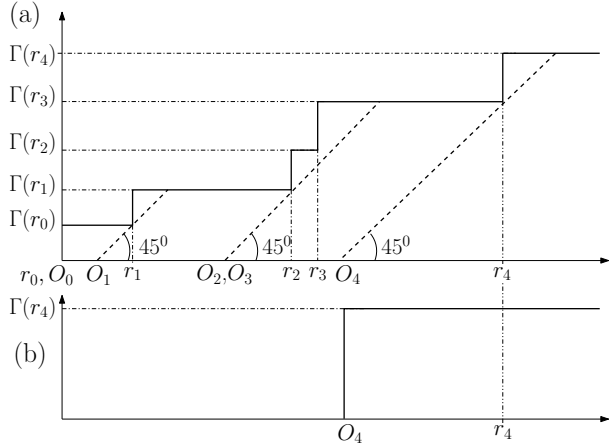}}
\caption{(a) Figure showing $O_i$'s which represent the first time instances at which the reciever can be kept on continuously for $\TRx(r_i)$ time. Note that $O_2$ and $O_3$ coincide in this example. (b) Energy harvesting profile at the receiver for problem $\mathsf{OFF}(O_4)$.}\label{figure_origin_points}
\end{figure}
Let \begin{equation}
i_0= \min \Bigg{\{}i: \lim_{t\rightarrow \infty}\TRx(r_i)g\left(\frac{\ETx(t)}{\TRx(r_i)}\right)\ge B_0\Bigg{\}},
\label{algo2_1}
\end{equation}
i.e. $i_0$ defines the earliest energy arrival time $r_i$ at the receiver such that the time ($\TRx(r_i)$) for which the receiver can stay on starting from $r_i$ is sufficient to transmit the $B_0$ bits by the transmitter eventually, even if no more energy arrives at the receiver. 

\begin{lemma}\label{lem:finitei0} If there is a solution to problem \eqref{pb2}, then $i_0< \infty$.
\end{lemma}

\begin{proof} Let the finish time of any feasible solution $\mathsf{F}$ for \eqref{pb2} be $T$. Then by time $T$, the maximum energy used by $\mathsf{F}$ to transmit $B_0$ bits at the transmitter is 
${\cal E}(T)$ and the receiver is on for at most time $\Gamma(T)$. Let the last energy arrival at the transmitter and the receiver before time $T$ be $\tau_{end}$ and $r_{end}$. Then, 
$i_0 \le \max\{\tau_{end}, r_{end}\}$, since starting from time $\max\{\tau_{end}, r_{end}\}$, one can transmit $B_0$ bits for function $g(.)$ using energy ${\cal E}(T)$ at the transmitter in receiver time of $\Gamma(T)$ without any break.
\end{proof}

Now, for the sake of applying algorithm $\mathsf{OFF}$ in multiple receiver energy arrivals regime, 
we introduce a new optimization problem, called $\mathsf{OFF}(O_i)$, for $i\ge i_0$.

$\mathsf{OFF}(O_i)$ is defined under the following energy harvesting profile - the receiver has only \textit{one} `time' arrival of $\TRx(r_i)$, the accumulated receiver time till $r_i$ in problem \eqref{pb2}, at time $O_i$ (see Fig. \ref{figure_origin_points} (b) for $i=4$). The transmitter energy harvesting profile remains same as $\ETx(t),\;\forall t\in [0,\infty)$. The formal description of problem $\mathsf{OFF}(O_i)$ is as follows.

\begin{align}
\min_{\{\bm{p},\bm{s},N\},T=s_{N+1}}	\;\;\;		& T \label{pb3}
\\
\text{subject to} 		\;\;\;		 B(T)&= B_0, 
\label{pb3_constraint_bits}
\\
     										 U(t)&\le \mathcal{E}(t)  		\;\;\;\;\;\;\;\;\;\;\;\;\; \forall \; t\;\in\;[O_i,T], \label{pb3_constraint_energy}
\\
    										C(t)& \leq \TRx(r_i) \;\;\;\;\;\;\;\;\;\;\;\forall \; t \; \in \; [O_i,T]. 
\label{pb3_constraint_time}
\\
    										C(t)& =0, U(t)=0  \;\;\forall \; t \; \in \; [0,O_i],
\label{pb3_constraint_origin}
\end{align}
where $C(t)$ is defined in \eqref{reciever_used}.

Since problem \eqref{pb3} has only one energy arrival at the receiver, we can use algorithm $\mathsf{OFF}$ to solve the problem of transmitting $B_0$ bits under this energy harvesting profile. With origin shifted to $O_i$, optimization problem \eqref{pb3} is similar to problem \eqref{pb1}.

From  Lemma \ref{lem:finitei0}, it is clear that if there is a solution to problem \eqref{pb2}, then $\forall\;\; i\ge i_0$ there is a solution for $\mathsf{OFF}(O_i)$. Let the optimal policy returned by solving $\mathsf{OFF}(O_i)$  be denoted by $X_i$. Moreover, its worthwhile remembering that $X_i$ is also a feasible solution to  \eqref{pb2}.
We have introduced $\mathsf{OFF}(O_i)$ to break the complex problem \eqref{pb2} into simpler single receiver `time' arrival problems $\mathsf{OFF}(O_i)$ that can be solved using $\mathsf{OFF}$. Lemma \ref{lemma_shifted_origin} states that the optimal solution to problem \eqref{pb2} is one of the $X_i$'s.

Also, following similar procedure as described in Lemma \ref{lemma_nobreaks}, we can show that there always exists an optimal solution to problem \eqref{pb2} with no breaks in transmission. So in the rest of the paper whenever we refer to the optimal solution of problem \eqref{pb2}, we consider the one without breaks in transmission.

\begin{lemma}
The optimal solution to problem \eqref{pb2} is policy $X_i$ for some $i$.
\label{lemma_shifted_origin}
\end{lemma}
\begin{proof}
We shall prove this by contradiction. Assume that the optimal solution to problem \eqref{pb2} is given by policy $Y,\{\bm{p},\bm{s},N\}$,  and none of the $X_i$'s are optimal to problem \eqref{pb2}. Let $O_k\le s_1<O_{k+1}$ for some $k$. 
By definition of $O_{k+1}$, all policies starting before $O_{k+1}$ must have transmission time less than  $\TRx(r_{k+1})$, and therefore the transmission time of $Y$ $(s_{N+1}-s_1)\le \TRx(r_{k+1}^-)=\TRx(r_{k})$. Let $X_k$ (solution of $\mathsf{OFF}(O_k)$) be denoted by $\{\bm{p}',\bm{s}',N'\}$. Now, since the transmission time of $Y$ is less than or equal to $\TRx(r_{k})$ and its start time is greater than $O_k$, policy $Y$ is a feasible solution to $\mathsf{OFF}(O_k)$. This implies $k \ge i_0$ and also,
\begin{equation}
s'_{N'+1}\le s_{N+1}.\label{l8_2}
\end{equation}
Both $X_k$ and $Y$ are feasible policies to problem \eqref{pb2}, and $Y$ is optimal to problem \eqref{pb2} from our assumption, whereas $X_k$ is not. Therefore, 
\begin{equation}
s_{N+1}<s'_{N'+1}.\label{l8_1}
\end{equation}
Hence, \eqref{l8_1} contradicts \eqref{l8_2}.
\end{proof}

From Lemma \ref{lemma_shifted_origin}, to solve \eqref{pb2} we need to identify the right index $i$ for which $X_i$ is optimal. Let the optimal policy for problem \eqref{pb2} be denoted by $X_{i^*}$. 
  Next, Lemma \ref{lemma_startpoint} states that if $X_{i^*}$, the optimal policy to $\mathsf{OFF}(O_{i^*})$, is also optimal to problem \eqref{pb2}, then it must begin transmission before $O_{i^*+1}$.

\begin{lemma}
The optimal policy to problem \eqref{pb2}, $X_{i^*}$ denoted by $\{\bm{p},\bm{s},N\}$, has $ s_1\le O_{i^*+1}$.
\label{lemma_startpoint}
\end{lemma}

\begin{proof}
We prove it by contradiction. 
Let $s_1>O_{i^*+1}$. Now, consider $X_{i^*+1}$ ($\{\bm{p'},\bm{s'},N'\}$), the optimal policy to $\mathsf{OFF}(O_{i^*+1})$ and $X_{i^*}$ ($\{\bm{p},\bm{s},N\}$). Since $s_{N+1}-s_1\le \TRx(r_{i^*})$ and $s_1>O_{i^*+1}$, $X_{i^*}$ is a feasible solution to $\mathsf{OFF}(O_{i^*+1})$. Now, both $X_{i^*}$ and $X_{i^*+1}$ are feasible solutions to $\mathsf{OFF}(O_{i^*+1})$ and $X_{i^*+1}$ is optimal with respect to $\mathsf{OFF}(O_{i^*+1})$. So, $s'_{N'+1}\le s_{N+1}$. 

But on the other hand, both $X_{i^*}$ and $X_{i^*+1}$ are feasible solutions to problem \eqref{pb2} and $X_{i^*}$ is optimal with respect to problem \eqref{pb2}. This implies that $s_{N+1}\le s'_{N'+1}$. From the above arguments we can conclude that the only possibility is $s_{N+1}= s'_{N'+1}$. 

So, both $X_{i^*}$ and $X_{i^*+1}$ are optimal with respect to $\mathsf{OFF}(O_{i^*+1})$. By Theorem \ref{th_algo1_1}, optimal solution to $\mathsf{OFF}(O_{i^*+1})$ is unique (optimal solution without breaks in transmission) and therefore, $X_{i^*}$ and $X_{i^*+1}$ have to be exactly identical. This would imply $s'_{1}=s_1>O_{i^*+1}$. Hence, by Lemma \ref{transmission_duration} on problem $\mathsf{OFF}(O_{i^*+1})$, $s'_{N'+1}-s'_{1}=\TRx(r_{i^*+1})$. Also, $s_{N+1}-s_1\le \TRx(r_{i^*})$ by the receiver `time' constraint in problem $\mathsf{OFF}(O_{i^*})$ and $\TRx(r_{i^*})<\TRx(r_{i^*})+\TRx_{i^*+1}=\TRx(r_{i^*+1})$. So, $s_{N+1}-s_1 < s'_{N'+1}-s'_{1}$ and this contradicts the identicality of policies $X_{i^*}$ and $X_{i^*+1}$.
\end{proof} 


Lemma \ref{lemma_first_solution} gives us a sufficient condition under which we can compute $X_{i^*}$. It establishes that the optimal policy to problem \eqref{pb2} is $X_{i^*}$, where $i^*$ is the minimum $i$ for which policy $X_i$'s start time is before $O_{i+1}$.
\begin{lemma}
\label{lemma_first_solution}
The optimal policy to problem \eqref{pb2} is $X_{i^*}$ where $$i^*= \min \;\{ i: s_1\le O_{i+1}, X_i\equiv\{\bm{p},\bm{s},N\}\}.$$
\end{lemma}
\begin{proof}
We will prove this by contradiction. Let $j$ denote the minimum $i$ for which policy $X_i$'s start time is less than or equal to $O_{i+1}$ and let $X_j$ be not optimal for problem \eqref{pb2}. $X_{i^*}$ being the optimal solution to problem \eqref{pb2}, satisfies Lemma \ref{lemma_startpoint} and  so $i^*>j$. Let $X_j$ be denoted by $\{\bm{p'},\bm{s'},N'\}$ and $X_{i^*}$ by $\{\bm{p},\bm{s},N\}$. Since $X_{i^*}$ is the optimal policy to problem \eqref{pb2}, we have 

\begin{equation}
 s_{N+1}\le s'_{N'+1}. 
 \label{lemma_first_solution_1}
\end{equation} 
Also, $s_{1}\ge O_{i^*}$ by definition of $X_{i^*}$, and $O_{i^*}\ge O_{j+1}$ since $i^*>j$ and $O_i$'s are non decreasing with respect to $i$. Moreover, $s'_1\le O_{j+1}$ by definition of $j$. This would imply $s'_1\le s_1$. Hence, using \eqref{lemma_first_solution_1}, we can write 
\begin{equation}
 s_{N+1}-s_1\le s'_{N'+1}-s'_1.
 \label{lemma_first_solution_2}
\end{equation} 
 From constraints of problem $\mathsf{OFF}(O_j)$, $s'_{N'+1}-s'_1\le \TRx(r_{j})$. Combining this with \eqref{lemma_first_solution_2}, $s_{N+1}-s_1\le \TRx(r_{j})$.
So, $X_{i^*}$ is a feasible solution to $\mathsf{OFF}(O_j)$. This would imply $s'_{N'+1}\le s_{N+1}$ on account of optimality of $X_j$ with respect to $\mathsf{OFF}(O_j)$. When combined with \eqref{lemma_first_solution_1}, we have $s_{N+1}= s'_{N'+1}$. Therefore, $X_j$, having same finish time with $X_{i^*}$, is also a optimal policy to problem \eqref{pb2}. But as we have shown earlier, the optimal policy is unique. Hence  we get a contradiction on our assumption on $X_j$. 
\end{proof} 
Now, we describe the algorithm $\mathsf{OFFM}$ to solve problem \eqref{pb2}.
\\
\textit{\textbf{Algorithm} $\mathsf{OFFM}$:}
\\
Initialization: Let $i=i_0$, where $i_0$ is defined in \eqref{algo2_1}.
\\
\textit{Step1:} Find policy $X_i$ as solution to $\mathsf{OFF}(O_i)$ using algorithm $\mathsf{OFF}$.
 \\
\textit{Step2:} If start time of $X_i$ is less than or equal to $O_{i+1}$ then output policy $X_i$ as the optimal policy and terminate.
If not, then increment $i$ to $i+1$ and go to \textit{Step1}.

Theorem \ref{th_multiple_energy} stated below establishes the optimality of algorithm $\mathsf{OFFM}$.
\begin{theorem}
Algorithm $\mathsf{OFFM}$ returns optimal solution to problem \eqref{pb2}.
\label{th_multiple_energy}
\end{theorem}
\begin{proof}
By Lemma \ref{lemma_shifted_origin}, we know that the solution to problem 2 has to be a policy $X_i$ for some $i$. Further, from Lemmas \ref{lemma_startpoint} and \ref{lemma_first_solution} we can conclude that the smallest index $i$ for which $X_i$ satisfies the condition of having its start time before $O_{i+1}$ is the optimal solution. 

As $X_{i^*}$ is the optimal solution to problem \eqref{pb2}, $i^* \ge i_0$, where $i_0$ is defined in \eqref{algo2_1}. Since $\mathsf{OFFM}$ iteratively finds policy $X_i$ for every value of $i\ge i_0$, it will definitely terminate with $X_{i^*}$ in less than $i^*$ number of iterations. Thus, Algorithm $\mathsf{OFFM}$ returns an optimal solution to problem \eqref{pb2}.
\end{proof}

{\it Discussion:} In this section, we derived the structure of the optimal power transmission profile in the offline setting, and derived an algorithm that satisfies the optimal structure. The main idea presented in this section is that the problem with multiple energy harvests at the receiver can be broken down into simpler problems, where there is only one energy harvest at the receiver. This hierarchical structure simplifies the complexity of the algorithm as well as provides us with an elegant method to construct a solution. As far as we know, such a hierarchical structure has not been discovered for other related energy harvesting problems.

%% file: onlineICC.tex
\label{sec:onlinereceiver}
In this section, we consider solving Problem \eqref{pb2}  in the more realistic online scenario, where the transmitter and the receiver 
are assumed to have only causal information about energy arrivals, and both have infinite battery capacities. To consider the most general model, even the  distribution of future energy arrivals is unknown at both the transmitter and the receiver.   

Let $B_{\mbox{\scriptsize{rem}}}(t)$ and $E_{\mbox{\scriptsize{rem}}}(t)$ denote the remaining number of bits to be transmitted, and energy left at the transmitter, at any time $t$, respectively, for the online algorithm. In place of $\{\bm{p},\bm{s},N\}$ for the offline case, we use the notation $\{\bm{l},\bm{b},M\}$ to denote an online algorithm, with identical definitions. Thus, $\bm{l}_i$ power is transmitted between time $\bm{b}_i$ and $\bm{b}_{i+1}$, and end time is $\bm{b}_{M+1}$. Let $\boldsymbol\sigma$ be the set of all possible energy arrival sequences at the transmitter, $\boldsymbol\rho$ be the set of all time arrival sequences at the receiver and $\mathbf{A}$ be the set of all online algorithms to solve \eqref{pb2}.  Then the  competitive ratio is given by 
\begin{equation}\label{defn:compratio}
{\mathsf r}=\displaystyle \min_{A\in \mathbf{A}} \max_{\sigma\in{\boldsymbol \sigma},\rho\in{\boldsymbol \rho}} \frac{T_A(\sigma,\rho)}{T_O(\sigma,\rho)},
\end{equation}
where $T_A(\sigma,\rho)$ and $T_O(\sigma,\rho)$ are the finish times taken by
the online algorithm $A$ and the optimal offline algorithm to Problem \eqref{pb2}, respectively.
Next, we present an online algorithm $\mathsf{ON}$ whose competitive ratio is strictly less than $2$, i.e., $$\max_{\sigma\in{\boldsymbol \sigma},\rho\in{\boldsymbol \rho}} \frac{T_{\mathsf{ON}}(\sigma,\rho)}{T_O(\sigma,\rho)} < 2$$

\textit{Online Algorithm $\mathsf{ON}$:} The algorithm waits till time $T_{\mbox{\scriptsize{start}}}$ 
which is the earliest energy arrival at transmitter or time addition at receiver such that using the energy $\ETx(T_{\mbox{\scriptsize{start}}})$ and time $\TRx(T_{\mbox{\scriptsize{start}}})$, $B_0$ or more bits can be transmitted, i.e.,
 
\begin{equation}
T_{\mbox{\scriptsize{start}}}=\min\ t \ s.t.\  \TRx(t)g\Bigg{(} \dfrac{\ETx(t)}{\TRx(t)}\Bigg{)}\ge B_0.\label{online_T_start}
\end{equation}

Starting at $T_{\mbox{\scriptsize{start}}}$, $\mathsf{ON}$ transmits with power $l_1$, such that $\frac{\ETx(T_{\mbox{\scriptsize{start}}})}{l_1}g(l_1)=B_0$.
After $T_{start}$, at \textit{every} energy arrival epoch $\tau_j$ of the transmitter, the transmission power is changed to $l_j$ such that

\begin{equation}
\frac{E_{\mbox{\scriptsize{rem}}}(\tau_j)}{l_j} g(l_j)= B_{\mbox{\scriptsize{rem}}}(\tau_j). 
\end{equation} 
Transmission power is not changed at any `time' arrival $r_j$ at the receiver after $T_{\mbox{\scriptsize{start}}}$, because there is sufficient receiver time already available to finish transmission. 

\begin{algorithm}
\caption {Online Algorithm $\mathsf{ON}$ for energy harvesting transmitter and receiver.}
\footnotesize
\label{algo_online}
\begin{algorithmic}[1]
\State \textbf{Input}: Bits to transmit $B_0$; $\ETx_i$, $\TRx_i$ for $\tau_i,r_i\le t$ where $t$ is the present time instant.

\State $T_{\mbox{\scriptsize{start}}}=\min\ t$ s.t. $\TRx(t)g\Bigg{(} \dfrac{\ETx(t)}{\TRx(t)}\Bigg{)}\ge B_0$
\State $B_{\mbox{\scriptsize{rem}}}=B_0$, $E_{\mbox{\scriptsize{rem}}}=\ETx(T_{\mbox{\scriptsize{start}}})$, $m=T_{\mbox{\scriptsize{start}}}$
\State Transmit at power $p$ such that $\dfrac{E_{\mbox{\scriptsize{rem}}}}{p} g(p)= B_{\mbox{\scriptsize{rem}}}$
\While {$t\le \left( m+\dfrac{E_{\mbox{\scriptsize{rem}}}}{p}\right)$}
	\If {$t=\tau_i$ for some $i$} 
		\State $B_{\mbox{\scriptsize{rem}}}=B_{\mbox{\scriptsize{rem}}}-(\tau_i-m)g(p)$
		\State $E_{\mbox{\scriptsize{rem}}}=E_{\mbox{\scriptsize{rem}}}+\ETx_i-(\tau_i-m)p$
		\State $m=\tau_i$
	\EndIf
	\State Transmit at power $p$ such that $\dfrac{E_{\mbox{\scriptsize{rem}}}}{p} g(p)= B_{\mbox{\scriptsize{rem}}}$

\EndWhile
\end{algorithmic}
\end{algorithm}

\textit{Example:} Fig. \ref{figure_online_example} shows the output of the proposed online algorithm $\mathsf{ON}$, \eqref{online_T_start} is not satisfied at time $\tau_0$, $r_1$, and $\tau_1$. At time $r_2$, \eqref{online_T_start} is satisfied and transmission starts with a power $l_1$ such that at rate $g(l_1)$, $B_0$ bits can be sent in $\ETx(r_2)/l_1$ time. Transmission power changes to $l_2$ at time $\tau_2$ such that $\frac{E_{\mbox{\scriptsize{rem}}}(\tau_2)}{l_2}g(l_2)=B_{\mbox{\scriptsize{rem}}}(\tau_2)$, and so on.
%
%
\begin{figure}
\centering
  	\centerline{\includegraphics[width=8cm]{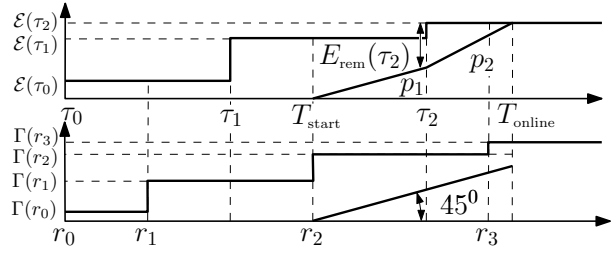}}

	\caption{An example for online algorithm $\mathsf{ON}$.}
	\label{figure_online_example}

\end{figure}
%
%
%
%
%

Next, we present certain properties of $\mathsf{ON}$ which would help us prove that it is $2$-competitive. Lemma \ref{online_power} proves that similar to the optimal offline algorithm (Lemma \ref{lemma_increasing_power}),  $\mathsf{ON}$ also has non-decreasing transmission powers. \begin{lemma}
The transmission powers are non-decreasing with time for $\mathsf{ON}$.
\label{online_power}

\end{lemma}
\begin{proof}
Combined with proof of Lemma \ref{lemma_online_inequality}.
\end{proof}
Lemma \ref{lemma_online_inequality} presented below is the key observation to proving Theorem \ref{thm:onlinecomp}. It helps provide a much shorter and elegant proof for competitive ratio less than $2$, compared to the proof presented in \cite{VazeEH2011} with no receiver constraints.
\begin{lemma}
If power transmitted by $\mathsf{ON}$ at time $t$ is $l$, then $\dfrac{\ETx(t)}{{l}}g(l) \le B_0,\;\;\forall\;\; t\in [T_{\mbox{\scriptsize{start}}}, T_{\mathsf{ON}}(\sigma, \rho)]$, with equality only at $t=T_{\mbox{\scriptsize{start}}}$.
\label{lemma_online_inequality}

\end{lemma}
\begin{proof}
After time $T_{start}$, power of $\mathsf{ON}$ is updated at each transmitter energy arrival epoch $\tau_j$. Hence, $b_i =\tau_j$ for some $j$, and $l_i$ and $\ETx(t)$ remains constant in $t\in[b_i,b_{i+1})$. Therefore, it is enough to prove that $\frac{g(l_i)}{l_i} \le \frac{B_0}{\ETx(b_i)}$ for $i\in\{1,\cdots,M\} $. We prove this by induction on $i \in \{1,2,\cdots,M\}$. 

With $b_1=T_{\mbox{\scriptsize{start}}}$, the base case follows since at time $T_{start}$, $\frac{\ETx(T_{\mbox{\scriptsize{start}}})}{l_1}g(l_1)=B_0$. Now, assume $\frac{g(l_{k-1})}{l_{k-1}}\le \frac{B_0}{\ETx(b_{k-1})}$ to be true for $k\in \{2,\cdots ,M\}$. As $b_k=\tau_j$ for some $j$,
\begin{align*}
\frac{l_{k}}{g(l_{k})}&=\frac{E_{\mbox{\scriptsize{rem}}}(b_{k})}{B_{\mbox{\scriptsize{rem}}}(b_{k})},
\\
&=\frac{E_{\mbox{\scriptsize{rem}}}(b_{k-1})-l_{k-1} (b_k-b_{k-1})+E_{j}}{B_{\mbox{\scriptsize{rem}}}(b_{k-1})-g(l_{k-1}) (b_k-b_{k-1})},
\\
&\stackrel{(a)}{=}\frac{l_{k-1}}{g(l_{k-1})}+\frac{E_{j}}{B_{\mbox{\scriptsize{rem}}}(b_{k-1})\gamma}
\\
&\stackrel{(b)}{>}\frac{\ETx(b_{k-1})}{B_0}+\frac{E_{j}}{B_0}
\\
&=\frac{{\ETx(b_{k})}}{B_0}.
\end{align*}
where $(a)$ follows from $\frac{B_{\mbox{\scriptsize{rem}}}(b_{k-1})}{E_{\mbox{\scriptsize{rem}}}(b_{k-1})}=\frac{g(l_{k-1})}{l_{k-1}}$ and defining $\gamma=\left( 1-\frac{l_{k-1}(b_k-b_{k-1})}{E_{\mbox{\scriptsize{rem}}}(b_{k-1})} \right)< 1$, $(b)$ uses induction hypothesis for the first term, along with $B_{\mbox{\scriptsize{rem}}}(b_{k-1})\gamma< B_0$ for the second term. This completes the proof of Lemma \ref{lemma_online_inequality}.
From $(a)$, we can see that $g(l_k)/l_k<g(l_{k-1})/l_{k-1}$. Hence, by monotonicity of $g(p)/p$, 
\begin{equation}
l_{k}>l_{k-1}, \;\;\forall k \in\{2,\cdots,M\},
\label{eq:online_monotonic_power}
\end{equation}
proving Lemma \ref{online_power}.
\end{proof}

Lemma \ref{online_time} establishes that the start time of $\mathsf{ON}$ must be earlier than the finish time of the optimal offline algorithm. Let $T_{\text{start}}(\sigma,\rho)$ be the starting time of $\mathsf{ON}$ for input $(\sigma,\rho)$.
\begin{lemma} 
With $\mathsf{ON}$, for any input $(\sigma,\rho)$, $T_{\text{start}}(\sigma,\rho) <T_{O}(\sigma,\rho)$.
\label{online_time}
\end{lemma}

\begin{proof}We prove this Lemma via contradiction. We fix an input $(\sigma,\rho)$ and show the result. We drop the suffix $(\sigma,\rho)$ for each of presentation. Suppose $T_{\mbox{\scriptsize{start}}} \ge T_{O}$. From \eqref{online_T_start}, either $T_{\mbox{\scriptsize{start}}}=\tau_i$ for some $i$ and/or $T_{\mbox{\scriptsize{start}}}=r_j$ for some $j$.
Let $T_{\mbox{\scriptsize{start}}}=\tau_i$. Since the optimal offline algorithm $\{\bm{p},\bm{s},N\}$ finishes before $T_{\mbox{\scriptsize{start}}}$ (which follows from our hypothesis), at the start time of the online algorithm,
the maximum (cumulative) energy utilized by the optimal offline algorithm $\{\bm{p},\bm{s},N\}$ is at most the energy arrived till time $T_{\mbox{\scriptsize{start}}}^-$. So, 
\begin{equation}
\sum_{i:p_i\neq 0}p_i(s_{i+1}-s_{i})\le \ETx(T_{\mbox{\scriptsize{start}}}^-)=\ETx(T_{\mbox{\scriptsize{start}}})-\ETx_i\neq \ETx(T_{\mbox{\scriptsize{start}}}).
\label{onilne:1}
\end{equation}
Similarly, if $T_{\mbox{\scriptsize{start}}}=r_j$, then the maximum time for which the receiver can be \textit{on} is
$\TRx(T_{\mbox{\scriptsize{start}}}^-)$. So, 
\begin{equation}\sum_{i:p_i\neq 0}(s_{i+1}-s_{i})\le\TRx(T_{\mbox{\scriptsize{start}}}^-)=\TRx(T_{\mbox{\scriptsize{start}}})-\TRx_j\neq \TRx(T_{\mbox{\scriptsize{start}}}).
\label{onilne:2}
\end{equation}


Therefore, the total number of bits transmitted by the optimal offline algorithm $\{\bm{p},\bm{s},N\}$ is given by

\begin{align}
\nonumber\sum_{i=1, p_i\neq 0}^{N} & g(p_i)(s_{i+1}-s_{i})
\\
&\nonumber \stackrel{(a)}{\le}g\left(\frac{\sum_{i:p_i\neq 0}p_i(s_{i+1}-s_{i})}{\sum_{j:p_j\neq 0}(s_{j+1}-s_{j})}\right)\sum_{j:p_j\neq 0} (s_{j+1}-s_{j}),
\\
&\nonumber\stackrel{(b)}\le g\left(\frac{\ETx(T_{\mbox{\scriptsize{start}}}^-)}{\TRx(T_{\mbox{\scriptsize{start}}}^-)}\right)\TRx(T_{\mbox{\scriptsize{start}}}^-)
\\
&\stackrel{(c)}{<}B_0,\label{online_eq_2}
\end{align}
where $(a)$ follows from Jensen's inequality since $g(p)$ is concave, $(b)$ follows from monotonicity of $g(p)/p$ and \eqref{onilne:1}, \eqref{onilne:2},
and $(c)$ follows from \eqref{online_T_start}. From \eqref{online_eq_2}, we can conclude that the optimal offline algorithm transmits $\sum_{i=1,\ p_i\neq 0}^{N} g(p_i)(s_{i+1}-s_{i})$ bits which is less than $B_0$, and therefore we arrive at a contradiction.
\end{proof}
Finally, Theorem \ref{thm:onlinecomp} proves that $\mathsf{ON}$ finishes strictly before twice the time taken by the optimal offline algorithm.
\begin{theorem}\label{thm:onlinecomp}
The competitive ratio of $\mathsf{ON}$ is less than $2$.
\end{theorem}
\begin{proof}




Let $\mathsf{ON}$ transmit with power $l_k$ at time $T_{O}^-$.
Since $T_{\mbox{\scriptsize{start}}}<T_{O}$ by Lemma \ref{online_time}, ${l}_k > 0$. Let $b_k<T_{O}$ be the time where transmission starts with power $l_k$. 
By definition, $\sum_{i=k}^{M}g(l_i)(b_{i+1}-b_i)=B_{\mbox{\scriptsize{rem}}}(b_k)$.
From  Lemma \ref{online_power},
\begin{equation}
(b_{N+1}-b_k)\le\frac{B_{\mbox{\scriptsize{rem}}}(b_k)}{g(l_k)}=
\frac{E_{\mbox{\scriptsize{rem}}}(b_k)}{l_k}\le \frac{\ETx(b_k)}{l_k}\le \frac{\ETx(T_{O}^-)}{l_k}.
\label{eq_online_time_1}  
\end{equation}
Applying Lemma \ref{lemma_online_inequality} at time $T_{O}^-$,
\begin{equation}
\frac{\ETx(T_{O}^-)}{l_k}g(l_k)\le B_0\stackrel{(a)}{\le}T_{O}\; g\left(\frac{\ETx(T_{O}^-)}{T_{O}}\right),
\label{eq_online_time_2}
\end{equation}
where $(a)$ holds because the maximum number of bits sent by the optimal offline algorithm by time $T_{O}$ can be bounded by $T_{O}\, g\left(\frac{\ETx(T_{O}^-)}{T_{O}}\right)$ due to concavity of $g(p)$.
By monotonicity of $g(p)/p$, from \eqref{eq_online_time_2}, it follows that  
$\frac{\ETx\left(T_{O}^-\right)}{l_k}\le T_{O}$.
Combining this with \eqref{eq_online_time_1},
$(b_{{N}+1}-b_k)\le T_{O}$.
As $b_k<T_{O}$, we calculate the competitive ratio as,

\begin{equation*}
\mathsf{r}=\max_{\sigma\in{\boldsymbol \sigma},\rho\in{\boldsymbol \rho}} \frac{T_{\mathsf{ON}}(\sigma,\rho)}{T_O(\sigma,\rho)} = \dfrac{(b_{{N}+1}-b_k)+b_k}{T_{O}} <  2.
\end{equation*}
\end{proof}

The next Theorem establishes that $\mathsf{ON}$ is an optimal online algorithm by showing that the competitive ratio of any online algorithm is arbitrarily close to $2$.
\begin{theorem}
 \label{thm:onlinelowerbound}
$\mathsf{ON}$ is an optimal online algorithm.
\end{theorem}
\begin{proof}
We will construct a set of two energy arrival sequences at the transmitter and the receiver for which the competitive ratio of any online algorithm is arbitrarily close to $2$ for at least one of the two sequences.

In order to calculate a lower bound ${\mathsf r}_{\ell}$ to ${\mathsf r}$ \eqref{defn:compratio}, we consider ${\boldsymbol \sigma}_s \subseteq {\boldsymbol \sigma} \text{ and }{\boldsymbol \rho}_s \subseteq {\boldsymbol \rho}$, a small subset of all possible energy harvesting (EH) sequences. Then, 
\begin{equation}
{\mathsf r} \ge {\mathsf r}_{\ell} = \displaystyle \min_{A\in \mathbf{A}} \max_{\sigma\in{\boldsymbol \sigma}_s,\rho\in{\boldsymbol \rho}_s} \frac{T_A(\sigma,\rho)}{T_O(\sigma,\rho)}.
\end{equation}

The idea behind the proof is to construct a set of two possible EH sequences ${\boldsymbol \sigma}_s = \{\sigma_1,\sigma_2\}$ at the transmitter with the same EH profile ${\boldsymbol \rho}_s = \{\rho_1\}$ at the receiver, where, with $\sigma_1$, the online algorithm $\mathsf{ON}$ provides a finish time ratio $\left(=\dfrac{T_{\mathsf{ON}}(\sigma_1,\rho_1)}{T_O(\sigma_1,\rho_1)}\right)$  of $1$, and with $\sigma_2$ it leads to a finish time ratio close to $2$. We then proceed to show that the minimum finish time for ${\boldsymbol \sigma}_s,{\boldsymbol \rho}_s$ over all algorithms in $\mathbf{A}$ is achieved by $\mathsf{ON}$. In doing so, we lower bound ${\mathsf r}_{\ell}$ by a value arbitrarily close to $2$. Combining this with the fact that ${\mathsf r}<2$ for $\mathsf{ON}$ (from Theorem \ref{thm:onlinecomp}), we can say that $\mathsf{ON}$ achieves the optimal competitive ratio. Now, it remains to show that ${\mathsf r}_{\ell}\ge 2^-$. 



Next, we explain the construction of ${\boldsymbol \sigma}_s$. Let $\sigma_1$ consist of only one EH arrival $ \ETx_0 $ at time $\tau_0=0$, and let $\sigma_2$ represent the EH sequence $\{\ETx_0, \ETx_1 \}$ occurring at time $\tau_0=0$ and $\tau_1=1$. We assume that the receiver has only one `time' arrival of $\TRx_0=T$ at time $t=0$, i.e. $\rho = \{T\}$ at time $r_0=0$.
Let $\ETx_0$  and $T>>0$ be chosen such that $B_0= Tg(\frac{\ETx_0}{T})$. Let $\ETx_1$ be such that $B_0=\tau_1g\left(\frac{\ETx_0+\ETx_1}{\tau_1}\right)$. 
The performance of algorithm $\mathsf{OFFM}$ and the online algorithm $\mathsf{ON}$ for energy arrival sequences $\{\sigma_1,\sigma_2\}$ is depicted in Fig. \ref{fig:onlinelowerbound}. Clearly, both $\mathsf{OFFM}$ and $\mathsf{ON}$ follow a constant power transmission policy for $\sigma_1$ where power $\frac{\ETx_0}{T}$ is transmitted from time $0$ to $T$. 
For $\sigma_2$, $\mathsf{OFFM}$ transmits with power $\frac{\ETx_0}{\tau_1}$ from time $0$ to $\tau_1$, and power $\frac{\ETx_1}{T_1-\tau_1}$ from $\tau_1$ to $T_1$, where $T_1$ is calculated by 
\begin{equation}\label{eq:Ctimesigma1}
(T_1-\tau_1)g\left(\frac{\ETx_1}{T_1-\tau_1}\right)=B_0-\tau_1g\left(\frac{\ETx_0}{\tau_1}\right).
\end{equation} 
Compared to this, with $\sigma_2$, $\mathsf{ON}$ transmits with power $\frac{\ETx_0}{T}$ from time $0$ to $\tau_1$ and power $\frac{\ETx_1+\ETx_0(1-\tau_1/T)}{T_2-\tau_1}$ from time $\tau_1$ to $T_2$, where $T_2$ is given by 
\begin{equation}
(T_2-\tau_1)g\left(\frac{\ETx_1+\ETx_0(1-\tau_1/T)}{T_2-\tau_1}\right)=B_0-\tau_1g\left(\frac{\ETx_0}{T}\right).
\end{equation}

Therefore, 
\begin{equation}\label{eq:compratioC}
\frac{T_{\mathsf{ON}}(\sigma_1,\rho_1)}{T_O(\sigma_1,\rho_1)}=1,\text{ while } \frac{T_{\mathsf{ON}}(\sigma_2,\rho_1)}{T_O(\sigma_2,\rho_1)}=\frac{T_2}{T_1}.
\end{equation}

Now consider any online algorithm $A\in \mathbf{A}$. Since, $A$ is assumed to use only causal information regarding energy harvests, it would generate identical power transmission profile with EH sequence $\sigma_1$ and $\sigma_2$ for time $0$ to $\tau_1$. Let $A$ use $\alpha$ fraction of energy $\ETx_0$ till time $\tau_1$.  Hence, we can characterize every online algorithm $A$ by $\alpha$, the fraction of energy it uses in time $[0, \tau_1]$,  and 
\begin{equation}
\mathsf{r}_{\ell} \ge \displaystyle \min_{\alpha \in [0,1]} \max_{\sigma \in \{ \sigma_1, \sigma_2\}} \frac{T_A(\sigma,\rho_1)}{T_O(\sigma,\rho_1)}.
\label{eq:r_l}
\end{equation}
Let us denote the corresponding value of $\alpha$ for algorithm $\mathsf{ON}$ as $\alpha'=1-\tau_1/T$.
The total receiver time being $\TRx_0=T$,
the maximum number of bits that can be transmitted  by any online algorithm $A$ with particular choice of $\alpha$, for EH sequence $\sigma_1$ is given by, 
\begin{equation}
B_{\alpha}=\tau_1g\left(\frac{\alpha\ETx_0}{\tau_1}\right)+(T-\tau_1)g\left(\frac{(1-\alpha)\ETx_0+\ETx_1}{T-\tau_1}\right).
\end{equation}
Because of the concavity of rate function $g$, from \eqref{eq:Ctimesigma1}, we can see that $B_{\alpha}\le B_0$, with equality iff $\alpha=\alpha'$. So, for $\alpha\neq \alpha'$, algorithm $A$ cannot transmit $B_0$ bits with EH sequence $\sigma_1$. Therefore, in the RHS \eqref{eq:r_l}, we only concern ourselves with the performance of online algorithm with $\alpha=\alpha'$, i.e. $\mathsf{ON}$. Hence,
\begin{align}
\mathsf{r}_{\ell} &\ge\max \left( \frac{T_{\mathsf{ON}}(\sigma_1,\rho_1)}{T_O(\sigma_1,\rho_1)},\frac{T_{\mathsf{ON}}(\sigma_2,\rho_1)}{T_O(\sigma_2,\rho_1)}\right)
\\
&=\frac{T_{\mathsf{ON}}(\sigma_2,\rho_1)}{T_O(\sigma_2,\rho_1)}=\frac{T_2}{T_1},
\end{align}
where the last equality follows from \eqref{eq:compratioC}.

Therefore, we only need to show that $\frac{T_2}{T_1}\ge 2^-$ for $\mathsf{ON}$ by choosing parameters $\ETx_0$ and $T$. It is difficult to obtain a closed form expression for $\frac{T_2}{T_1}$ in terms of relevant parameters, hence we lower bound $\frac{T_2}{T_1}$ by constructing an example sequence $\{ \sigma_1,\sigma_2, \rho_1\}$ as follows.
With $\ETx_0=10^{-4}$, $T=10^4$,  and $g(p)=0.5\log_2(1+p)$, we get $\frac{T_2}{T_1}=2-2.49\times 10^{-4}$. Similarly, by increasing $T$ and decreasing $\ETx_0$ towards $0$, we can keep pushing $\frac{T_2}{T_1}$ arbitrarily close to $2$. This completes the proof.
\end{proof}
{\it Discussion:} In this section, we derived an optimal online algorithm when EH is employed at both the transmitter and the receiver. First, we proposed an online algorithm and showed that it finishes the transmission of required number of bits in at most twice the time an optimal offline algorithm takes  knowing all energy arrivals non-causally. Moreover, the online algorithm is independent of the energy arrival distributions both at the transmitter and the receiver, so has built-in robustness. Also, note that the proof of Theorem \ref{thm:onlinecomp} does not explicitly require to know the exact structure of the optimal offline algorithm. Thereafter, to complete the characterization of optimal online algorithms, we showed that no online algorithm can do better than the proposed online algorithm by constructing a set of energy arrival sequences for which any online algorithm will have competitive ratio arbitrarily close to two for at least one of the energy arrival sequences. Typically, finding a (tight) lower bound on the competitive ratio for all online algorithms is a hard problem, but we are able to accomplish this for the transmission finish time minimization problem.  

After examining the case of EH being employed at both the transmitter and the receiver with no battery constraint until now in this paper, we next consider the more reasonable model of a finite battery availability at both the transmitter and the receiver, and derive online algorithm with bounded competitive ratio.
\begin{figure}
\centering
  	\centerline{\includegraphics[width=8cm]{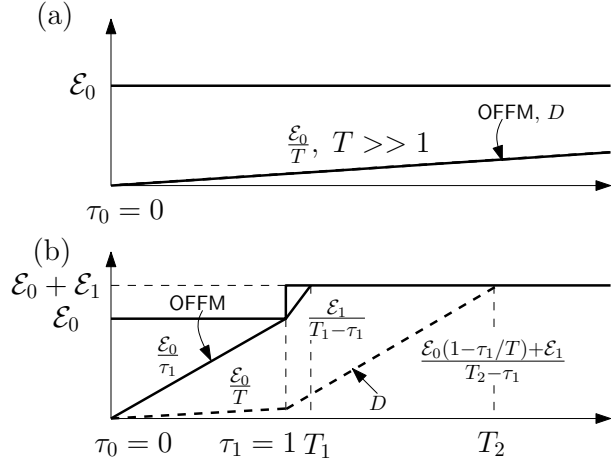}}
	\caption{Transmission policy of $\mathsf{OFFM}$ and the online algorithm $\mathsf{ON}$ for EH profile (a) $\sigma_1$  and (b) $\sigma_2$.}
	\label{fig:onlinelowerbound}
\end{figure}

%% file: online_battery.tex
In previous sections, an infinite battery capacity was assumed at both the transmitter and the receiver. In this section, to make the discussion more practical, we consider the case when both the transmitter and receiver battery have finite capacity. Also, we consider the online setting for obvious practical reasons.

 \subsection{EH only at the transmitter}
 \label{sec:finbatonlinetrans}
For ease of exposition, we first discuss the finite battery model where only the transmitter is EH powered, while the receiver is powered by a conventional power source. We extend the analysis to include an EH powered receiver in Section \ref{sec:finbatonlinereceiver}. 
With finite battery capacity, similar to Section \ref{sec:onlinereceiver}, under a worst case input for energy harvests, the online algorithm might not finish transmission of $B_0$ bits ever, while an offline algorithm can, making the competitive ratio infinity. Thus, we consider the non-degenerate online setting, where the amount of energy arriving at any instant is a random variable whose probability density function (PDF) $f(x)$ is known ahead of time. Note that on the realization basis, only causal information is revealed to any online algorithm. 

For simplicity, we divide time into slots of length $w$, with $\ETx_i$ amount of energy arriving in the $i^{th}$ slot. We also assume that energy is harvested at the beginning of the  slot. The amount of energy arriving in any slot $i$, $\ETx_i$, is assumed to follow an i.i.d. PDF $f(x)$ for all $i\ge 0$. The transmitter is assumed to have a  battery capacity $\mathcal{C}_t$. Thus without loss of generality we assume that $f(x)=0$ for $x>\mathcal{C}_t$.  Following Problem \eqref{pb1}, we want to transmit $B_0$ bits in total in minimum time under this online setting with finite battery capacity at the transmitter. The system model is shown in Fig. \ref{online_batt} (a). With randomized energy inputs, we consider the expected competitive ratio as the performance metric to design online algorithms, that is defined as the expectation of the ratio of the  time taken by an online algorithm and the time taken by an optimal offline algorithm. 
We next present an online algorithm which we call Accumulate\&Dump to upper bound the expected competitive ratio.

\textit{Algorithm Accumulate\&Dump:} In the first iteration, algorithm Accumulate\&Dump waits for $\mathcal{N}$ slots such that at least $\mathcal{C}_t/c$ amount of battery capacity is filled, where $c\ge 1$ is a positive constant in the algorithm. The value  $c$ is dependent on $f(x)$ and we will calculate the best choice of $c$ for a given distribution $f(x)$ while analysing the algorithm. 
Clearly $\mathcal{N}$ is a random variable given by, 
\begin{align}
&\mathcal{N}=\min \left\{ n : \sum_{i=0}^{n-1} \ETx_i\ge \frac{\mathcal{C}_t}{c}\right\}.
\label{min_slots}
\end{align}
After accumulating at least $\mathcal{C}_t/c$ amount of energy, Accumulate\&Dump uses all of the available energy $\sum_{i=0}^{\mathcal{N}-1} \ETx_i$ in the battery with a constant rate in $w$ amount of time i.e. within the $\mathcal{N}^{th}$ slot. With an empty battery at the end of $\mathcal{N}^{th}$ slot, the transmitter starts accumulating energy afresh and continues the above process until it transmits all $B_0$ amount of bits.
\\
\textit{Example:} Fig.\ref{online_batt} (b) shows an example of running Accumulate\&Dump with $\mathcal{N}=2$ in the first iteration and $\mathcal{N}=3$ in the next.
\begin{figure}
\centering
  	\centerline{\includegraphics[width=8cm]{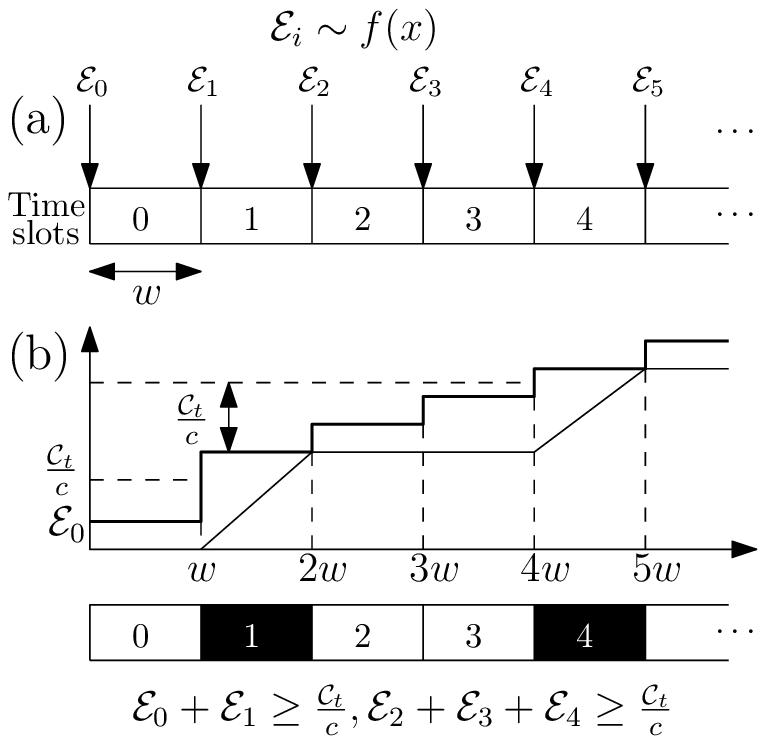}}
	\caption{(a) Transmitter model for slotted energy arrival (b) An example for Accumulate\&Dump algorithm.}
	\label{online_batt}
\end{figure}
\\
\textit{Analysis:} Consider the sum process of the i.i.d. random variables $\ETx_i$, $\sum_{i=0}^{\mathcal{N}-1} \ETx_i$, where 
$\mathcal{N}$ defined in \eqref{min_slots}. Let us denote condition $H$ as
\begin{equation}
H\equiv\left(\sum_{i=0}^{\mathcal{N}-1} \ETx_i\ge \mathcal{C}_t/c\;\text{ and } \;\sum_{i=0}^{\mathcal{N}-2} \ETx_i< \mathcal{C}_t/c\right).
\label{this_is_H}
\end{equation}
Note that the stopping condition in \eqref{min_slots} is equivalent to $H$, as $f(x)=0$ for $x<0$. Lemma \ref{lemma_walds} formulates an expression for the expected value of $\mathcal{N}$. First, we state the form of Walds' equation \cite{grimmett1985probability} that we use in Lemma \ref{walds_equation}.
\begin{lemma}
$\left[\textbf{Walds equation}\right]$ 
If $\mathcal{S}$ is a stopping time with respect to an i.i.d. sequence
$\{X_n : n \ge 1\}$, and if $\mathbf{E}[\mathcal{S}] < \infty $ and $\mathbf{E}[|X_i|] < \infty$, then $\mathbf{E}[\sum_{i=1}^\mathcal{S} X_i]=\mathbf{E}[\mathcal{S}]\mathbf{E}[X_i]$.
\label{walds_equation}
\end{lemma}
\begin{lemma}
 \label{lemma_walds}
 $\mathbf{E}[\mathcal{N}]=\dfrac{\mathbf{E}\left[\sum_{i=0}^{\mathcal{N}-1} \ETx_i \right]}{\mathbf{E}[\ETx_0]}$.\end{lemma}
\begin{proof} Clearly, $\mathcal{N}$ is a stopping time. The proof directly follows from Lemma \ref{walds_equation}, once we show that $\mathbf{E}[\mathcal{N}]$ is finite that is proved in Appendix \ref{AppendixA}.
\end{proof}

Next, we analyze the competitive ratio of the  online algorithm Accumulate\&Dump. Since, $\ETx_i$'s are random variables, we use expected competitive ratio analysis for Accumulate\&Dump, and prove an upper bound to it in Theorem \ref{online_bat_trans}. In doing so, we primarily consider a class of distributions that satisfy the following condition.

\begin{assumption}
$\mathbf{E}[\ETx_0|\ETx_0\ge \gamma]\le \gamma+\mathbf{E}[\ETx_0],\;\; \forall \gamma \in [0,\mathcal{C}_t).$\label{new_condition}
\end{assumption} This assumption simply means that the expected jump size given that it is larger than $\gamma$ is no more than if the origin is shifted to $\gamma$ and the process takes an i.i.d. jump from there. 
Note that most light-tailed continuous distributions satisfy Assumption \ref{new_condition}. For example, uniform distribution satisfies Assumption \ref{new_condition} with a strict inequality, while the exponential distribution satisfies Assumption \ref{new_condition} with an equality.
\begin{remark} To find a bound on $\mathbf{E}[\mathcal{N}]$, we need an upper bound on 
$\mathbf{E}\left[\sum_{i=0}^{\mathcal{N}-1} \ETx_i \right]$. Towards that end, we need a bound on the expected value of the $\mathcal{N}^{th}$ increment of process $\ETx_i$ given that $\sum_{i=0}^{\mathcal{N}-1} \ETx_i > \frac{\mathcal{C}_t}{c}$.
Without Assumption \ref{new_condition},  knowing that $\sum_{i=0}^{\mathcal{N}-1} \ETx_i > \mathcal{C}_t/c$, there is no easy way of bounding the value ${\mathbf E}\left[\sum_{i=0}^{\mathcal{N}-1} \ETx_i\right]$ except of course the trivial bound of $\mathcal{C}_t$. For example, suppose $\mathcal{E}_i$'s have a Bernoulli distribution over 
$\{0,x\}$ with probability $\{ p, 1-p\}$. If $x$ is large, and we condition on $x> 0$, then $\mathbf{E}[x] = x$. Heavy tailed distributions also do not allow any bound on the expected value of ${\mathbf E}\left[\sum_{i=0}^{\mathcal{N}} \ETx_i\right]$ at the cross-over point. As we will see in Proof of Theorem \ref{online_bat_trans}, Assumption \ref{new_condition} is sufficient to obtain non-trivial bound on the expected jump size given that the jump is larger than a certain threshold. 
\end{remark}
\begin{theorem}
The expected competitive ratio of Accumulate\&Dump algorithm is finite.
\label{online_bat_trans}
\end{theorem}
\begin{proof}
Let the number of slots taken by the optimal offline algorithm to finish transmitting $B_0$ bits is $\mathcal{S}_{\text{off}}$ and the number of slots taken by Accumulate\&Dump to complete is $\mathcal{S}_{on}$. Since the maximum amount of energy harvested in one slot is bounded by $\mathcal{C}_t$, to bound $\mathcal{S}_{\text{off}}$, we consider the best case scenario (that is the fastest completion of transmission) where $\mathcal{C}_t$ amount of energy arrives in each slot. In this best case, the number of bits transmitted per slot, i.e. slot width of $w$, is $wg(\mathcal{C}_t/w)$, and the total number of slots taken to transmit $B_0$ bits is $\Bigg{\lceil}\dfrac{B_0}{wg(\mathcal{C}_t/w)}\Bigg{\rceil}$. Therefore, 
\begin{equation}
\mathcal{S}_{\text{off}}\ge \frac{B_0}{wg(\mathcal{C}_t/w)}.
\end{equation} 
Now, we can write the expected competitive ratio as, 
\begin{align}
&\mathbf{E}[{\mathsf r}]=\mathbf{E}\left[\frac{\mathcal{S}_{on}}{\mathcal{S}_{\text{off}}}\right]\le \mathbf{E}\left[\frac{\mathcal{S}_{on}}{\frac{B_0}{wg(\mathcal{C}_t/w)}}\right]=\frac{\mathbf{E}\left[\mathcal{S}_{on}\right]}{\frac{B_0}{wg(\mathcal{C}_t/w)}}.
\label{competitve_dist}
\end{align}
In each iteration, Accumulate\&Dump waits for $(\mathcal{N}-1)$ slots by which time 
it accumulates at least $\mathcal{C}_t/c$ amount of energy, and then uses all the accumulated energy in the $\mathcal{N}^{th}$ slot for transmission. Hence, at least $wg\left(\frac{\mathcal{C}_t}{cw}\right)$ bits are transmitted in the $\mathcal{N}^{th}$ slot by Accumulate\&Dump, where $\mathcal{N}$  is defined in \eqref{min_slots}. Note that $\mathcal{N}$ is  i.i.d. random variable over all iterations of Accumulate\&Dump. Thus the number of bits transmitted by Accumulate\&Dump in time $\mathcal{N}w$  is $wg\left(\frac{\mathcal{C}_t}{cw}\right)$. This implies that the maximum number of iterations (say $m$) taken by Accumulate\&Dump to transmit $B_0$ bits is $\Bigg{\lceil}\dfrac{B_0}{wg\left(\frac{\mathcal{C}_t}{cw}\right)}\Bigg{\rceil}$. Hence,
\begin{align}
\mathbf{E}\left[\mathcal{S}_{on}\right]&=\mathbf{E}[\mathcal{N}\times m]
\\
&\le \mathbf{E}\left[\mathcal{N}\right]\Bigg{\lceil}\dfrac{B_0}{wg\left(\frac{\mathcal{C}_t}{cw}\right)}\Bigg{\rceil}\stackrel{(a)}{\approx} \mathbf{E}\left[\mathcal{N}\right]\dfrac{B_0}{wg\left(\frac{\mathcal{C}_t}{cw}\right)},
\label{n_on}
\end{align}
where $(a)$ follows under the assumption that $m>>1$. 

Under Assumption \ref{new_condition}, as shown in Appendix \ref{AppendixB},
 \begin{equation}\label{eq:jumpspecial}
\mathbf{E}[\mathcal{N}]\le\frac{\mathcal{C}_t/c}{\mathbf{E}[\ETx_0]}+1.
\end{equation}

 Without assumption \ref{new_condition}, using the trivial upper bound 
 $\mathbf{E}\left[\sum_{i=0}^{\mathcal{N}-1} \ETx_i \right] \le \mathcal{C}_t$, we get 

 \begin{equation}\label{eq:jumpgeneral}
\mathbf{E}[\mathcal{N}]\le\frac{\mathcal{C}_t/c+\mathcal{C}_t}{\mathbf{E}[\ETx_0]}
\end{equation}
 for any general distribution $f(x)$, which is also shown in Appendix \ref{AppendixB}.

Thus, from \eqref{competitve_dist},  \eqref{n_on}, and \eqref{eq:jumpspecial}, we can write the competitive ratio under Assumption \ref{new_condition} as,
\begin{align}
\nonumber \mathbf{E}[{\mathsf r}]&\le \frac{\left(\dfrac{\mathcal{C}_t/c}{\mathbf{E}[\ETx_0]}+1
   \right) \dfrac{B_0}{wg\left(\frac{\mathcal{C}_t}{cw}\right)}}{\frac{B_0}{wg\left(\mathcal{C}_t/w\right)}},
\\
&=\left(\dfrac{\mathcal{C}_t/c}{\mathbf{E}[\ETx_0]}+1\right) \dfrac{g\left(\frac{\mathcal{C}_t}{w}\right)}{g\left(\frac{\mathcal{C}_t}{cw}\right)}.
\label{upper_bound}
\end{align}

Recall that we can choose the parameter $c$. For $c=\dfrac{\mathcal{C}_t}{\mathbf{E}[\ETx_0]}$, \eqref{upper_bound} reduces to $\mathbf{E}[{\mathsf r}] \le \dfrac{2 g\left(\frac{\mathcal{C}_t}{w}\right)}{g\left(\frac{\mathbf{E}[\ETx_0]}{w}\right)}$. With $g(p)=0.5 \log_2(1+p)$, $\dfrac{g\left(\frac{\mathcal{C}_t}{w}\right)}{g\left(\frac{\mathbf{E}[\ETx_0]}{w}\right)}$ is constant for any distribution $f(x)$, where $\mathcal{C}_t$ scales polynomially with $\mathbf{E}[\ETx_0]$. Thus, we get a constant upper bound for $\mathbf{E}[{\mathsf r}]$ under Assumption \ref{new_condition}. 

For any general distribution $f(x)$, from \eqref{eq:jumpgeneral}, 
$$\mathbf{E}[{\mathsf r}]\le \left(\dfrac{\mathcal{C}_t/c+\mathcal{C}_t}{\mathbf{E}[\ETx_0]}\right) \dfrac{g\left(\frac{\mathcal{C}_t}{w}\right)}{g\left(\frac{\mathcal{C}_t}{cw}\right)},$$ which can shown to be finite for $c = \dfrac{\mathcal{C}_t}{\mathbf{E}[\ETx_0]}$, as above, but now the bound depends on system parameters $\mathcal{C}_t$ and $\mathbf{E}[\ETx_0]$.
\end{proof}
Next, we evaluate the derived bounds for particular energy arrival distributions.
\begin{example}
For uniform distribution $f(x)=1/\mathcal{C}_t,0\le x\le \mathcal{C}_t$, we have $\mathbf{E}[\ETx_0]=\mathcal{C}_t/2$ and we can reduce \eqref{upper_bound} to $\mathbf{E}[{\mathsf r}] \le 2\dfrac{\log_2\left(1+\frac{\mathcal{C}_t}{w}\right)}{\log_2\left(1+\frac{\mathcal{C}_t}{2w}\right)}< 4$.
\label{example1}
\end{example}
\begin{example} 
For exponential energy arrival at the transmitter, we can write the probability density function $f(x)$ as, 
\begin{align}
f(x)&=\lambda e^{-\lambda x}, && 0\le x<\mathcal{C}_t,
\\
&=e^{-\lambda \mathcal{C}_t}, && x=\mathcal{C}_t,
\\
&=0, && x>\mathcal{C}_t, x<0.
\end{align}

Let us assume that the probability of the energy arrival being more than the battery capacity is given by $10^{-\epsilon}=f(\mathcal{C}_t)$ for $\epsilon>0$. Note that $\mathbf{E}[\ETx_0]=(1-10^{-\epsilon})/\lambda$. So, with $c=\frac{\mathcal{C}_t}{\mathbf{E}[\ETx_0]}$, the upper bound \eqref{upper_bound} reduces to
\begin{align}
 2\dfrac{\log_2\left(1+\frac{ \epsilon\ln 10}{\lambda w}\right)}{\log_2\left(1+\frac{(1-10^{-\epsilon})}{\lambda w}\right)}&<2\max\left( 1,\frac{\epsilon\ln 10}{1-10^{-\epsilon}}\right) 
\label{exp_upper_boundexact} 
 \\
 & \approx 2 ,\;\;\;\;\;\;\;\;\;\;\;\text{ for }\epsilon<0.43 \\
 &  \approx 4.6 \epsilon, \;\;\;\;\;\;\;\text{ for }\epsilon\ge 0.43. 
 \label{exp_upper_bound}
\end{align}
\label{example2}
\vspace{-1cm}
\end{example}

\subsection{EH at both transmitter and receiver:}
\label{sec:finbatonlinereceiver}
After analyzing the expected competitive ratio when only the transmitter is powered by EH, in this subsection, 
we generalize expected competitive ratio analysis of subsection \ref{sec:finbatonlinetrans} to allow for both 
transmitter and receiver to be powered by EH and where both have finite battery capacities. The transmitter model remains as defined in subsection \ref{sec:finbatonlinetrans}, while for the receiver we assume that 
$\cR_i$, the energy arriving at each slot, is i.i.d. with PDF $h(x)$. The receiver has a finite battery capacity $\mathcal{C}_r$ and it uses $P_{r}$ amount of power to be \textit{on}. Also, $h(x)=0$ for $x>\mathcal{C}_r$ and $x<0$.

In this model, we propose a natural extension of Accumulate\&Dump as follows, \textit{Algorithm modified Accumulate\&Dump:}  In the first iteration, the algorithm waits for $\mathcal{N}$ energy arrivals such that at least $\mathcal{C}_t/c$ amount of energy is harvested at the transmitter, and at least $P_rw$ amount is accumulated at the receiver, where $c\ge 1$. That is, 
\begin{align}
&\mathcal{N}=\min \left\{ n : \sum_{i=0}^{n-1} \ETx_i\ge \frac{\mathcal{C}_t}{c} \text{ and } \sum_{i=0}^{n-1} \cR_i\ge P_rw\right\}.
\label{min_slots_rec}
\end{align}
In the $\mathcal{N}^{th}$ slot, the transmitter uses all the accumulated energy to transmit at a constant rate, and the receiver is also \textit{on}. After this, the system is essentially reset and the algorithm proceeds to the next iteration.

For this modified Accumulate\&Dump algorithm we provide a expected competitive ratio bound in Theorem \ref{online_bat_transrec}.
\begin{theorem}
The expected competitive ratio when both transmitter and receiver are powered by EH is upper bounded by
\begin{align*}
\mathbf{E}[{\mathsf r}] \le \Bigg{(}\dfrac{\mathcal{C}_r+P_rw}{\mathbf{E}[\cR_0]}&+\dfrac{\mathcal{C}_t+\mathcal{C}_t/c}{\mathbf{E}[\ETx_0]}\Bigg{)} \dfrac{g(\frac{\mathcal{C}_t}{w})}{g\left(\frac{\mathcal{C}_t}{cw}\right)},
\end{align*}
for any general distribution of $f(x)$ and $h(x)$, and 
\begin{align*}
\mathbf{E}[{\mathsf r}] \le \Bigg{(}\dfrac{P_rw}{\mathbf{E}[\cR_0]}&+\dfrac{\mathcal{C}_t/c}{\mathbf{E}[\ETx_0]}+2\Bigg{)} \dfrac{g(\frac{\mathcal{C}_t}{w})}{g\left(\frac{\mathcal{C}_t}{cw}\right)},
\end{align*}
when $f(x)$ and $h(x)$ satisfy Assumption \ref{new_condition}.
\label{online_bat_transrec}

\end{theorem}
\begin{proof}
Let us define $\mathcal{N}'$ and $\mathcal{N}''$ as 
\begin{align*}
&\mathcal{N}'=\min\left\{ n :  \sum_{i=0}^{n-1} \ETx_i\ge \frac{\mathcal{C}_t}{c}\right\},
\\
&\mathcal{N}''=\min\left\{ n :  \sum_{i=0}^{n-1} \cR_i\ge P_rw\right\}.
\label{min_slots_1}
\end{align*}
Hence, from \eqref{min_slots_rec}, $\mathcal{N}=\max (\mathcal{N}',\mathcal{N}'')$. Using $\mathbf{E}[\mathcal{N}]< \mathbf{E}[\mathcal{N}']+\mathbf{E}[\mathcal{N}'']$, the rest of the proof follows along similar lines as the proof of Theorem \ref{online_bat_trans}.
\end{proof}

{\it Discussion:} In this section, we proposed a simple online algorithm that accumulates energy upto a certain threshold and transmits (dumps) all of it as soon as it crosses the threshold. The chosen threshold controls the rate at which algorithm transmits power, larger the threshold less slots are active but with more power and vice versa. The idea behind this algorithm is that given that the rate function is concave (e.g., $\log$), the effect of not transmitting power in every slot is not too large, and one can tradeoff the threshold appropriately to find the best threshold given the energy arrival distribution information. We show that for most `nice' energy arrival distributions that do not have arbitrarily large jumps, we can bound the expected competitive ratio by a constant that does not depend on the system parameters, thus showing that the proposed algorithm is close to optimal and has reasonable performance.

%% file: experiments.tex

In this section, we first present a sample run of algorithm $\mathsf{OFFM}$ with $B_0=1$ as shown in Fig. \ref{figure_OFFM_example} for a given energy arrival sequence at the transmitter and receiver.
We can see from the transmitter-receiver energy profiles that the finish time decreases through policy $X_1$ to $X_4$. At the same, the transmission time increases from policy $X_1$ to $X_4$. Among policies $X_i$, $i=1,2,3,4$,  $X_4$ is the first policy whose start time is less than $O_{i+1}$(not shown in Fig. \ref{figure_OFFM_example}). Hence, by Lemma \ref{lemma_first_solution}, $X_4$ is optimal.


 Next, we perform simulations to illustrate the competitive ratio performance of the online algorithm from section \ref{sec:onlinereceiver} with no battery capacity constraints. The amount of energy harvested at the transmitter, and the energy (or time) harvested at the receiver are drawn from a uniform distribution in $[0,1]$. The inter-arrival distribution of energy harvests at transmitter and receiver is uniform in $[0,1]$. The rate function is assumed to be $g(p)=0.5\log_2(1+p)$.  Comparison between the online algorithm and $\mathsf{OFFM}$ is shown in Fig. \ref{figure_offon_example}. We can observe that the competitive ratio is close to $1.5$ for different values of $B_0$ bits, which is far better than the worst case bound of $2$ as derived in Theorem \ref{thm:onlinecomp}.

\begin{figure*}
\centering
  	\centerline{\includegraphics[width=14cm]{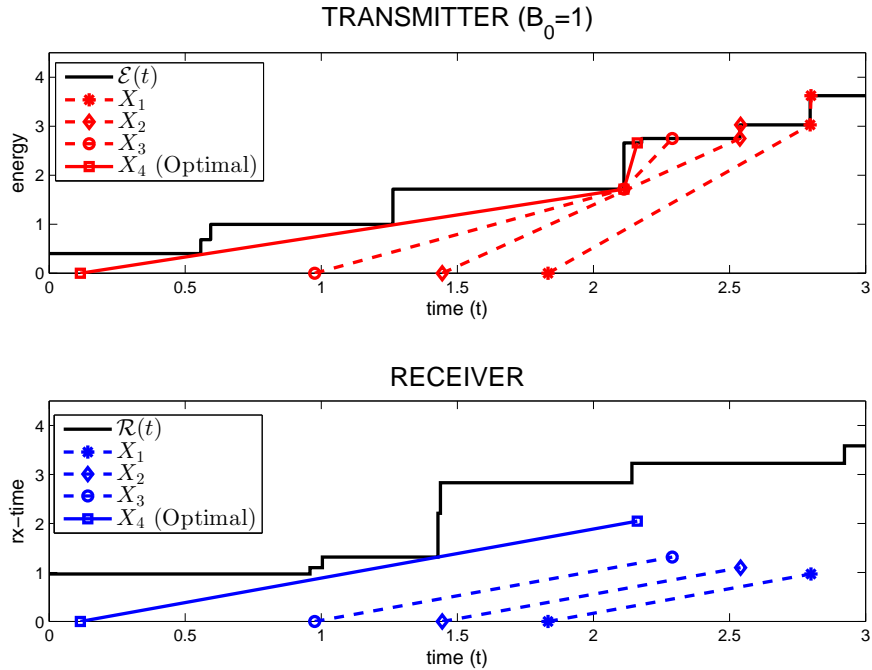}}

	\caption{An example for $\mathsf{OFFM}$ algorithm.}
	\label{figure_OFFM_example}

\end{figure*}
\begin{figure}
\centering
\centerline{\includegraphics[width=8cm]{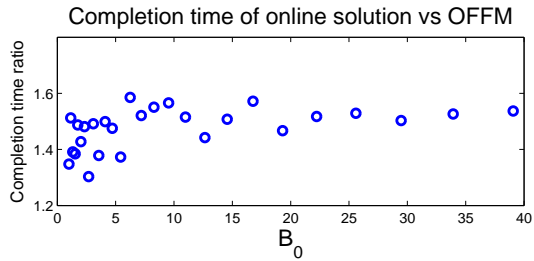}}
	\caption{$\mathsf{OFFM}$ algorithm vs online algorithm with no battery capacity constraints.}
	\label{figure_offon_example}

\end{figure}
%
%

For the finite battery setting, in Fig. \ref{figure_finitebat_example} (a), we first simulate the case when only the transmitter is powered by EH, where energy arrivals follow an exponential distribution, and demonstrate the competitive ratio of Accumulate\&Dump compared to the optimal offline algorithm \cite{Yener2012optbat}. In this experiment, we assume the rate function to be given by $g(p)=0.5\log_2(1+p)$, the battery capacity to be $115$ units, the slot width to be $5$ units, and the energy arrival distribution to be exponential with mean $25$. As described in Example \ref{example2}, the distribution is truncated, i.e. any energy arrival of amount more that $\mathcal{C}_t$ is assumed to have a value of exactly $\mathcal{C}_t$. We have chosen our battery capacity so that the value for $\epsilon$ comes out to be $2$. That is, there is a probability of $0.01$ that the energy harvested is more that $\mathcal{C}_t$. Minimizing the upper bound on the expected competitive ratio given in \eqref{upper_bound}, the optimal value comes out to be $3.56$ for $c=5.07$. Although the theoretical upper bound computed is $3.56$, we can see that the simulated competitive ratio converges around $1.27$. 

Then we consider the case when both transmitter and receiver are powered with EH in the finite battery setting, and simulate the competitive ratio in Fig. \ref{figure_finitebat_example} (b). In this model, both the transmitter and the receiver are assumed to harvest energy from   exponential distribution with mean $25$, and both have a battery capacity of $115$. The receiver \textit{on} power is assumed to be  $P_r=7$. With $w=5$ and $c=5.07$, we can see that the upper bound on expected competitive ratio calculated using Theorem  \ref{online_bat_transrec} turns out to be $8$. An important point to note here is that the optimal offline algorithm is not known for the  finite battery setting when both the transmitter and the receiver are powered by EH. Thus, to compute the competitive ratio, we consider the optimal offline algorithm \cite{Yener2012optbat} when only the transmitter is powered by EH with  finite battery, which clearly is an upper bound on the performance when both transmitter and receiver are powered by EH.
In simulations, we compare the modified Accumulate\&Dump algorithm from section \ref{sec:finbatonlinereceiver} with the optimal offline algorithm presented in \cite{Yener2012optbat}. The simulated competitive ratio converges around $1.65$. 



\begin{figure}
\centering
  	\centerline{\includegraphics[width=8cm]{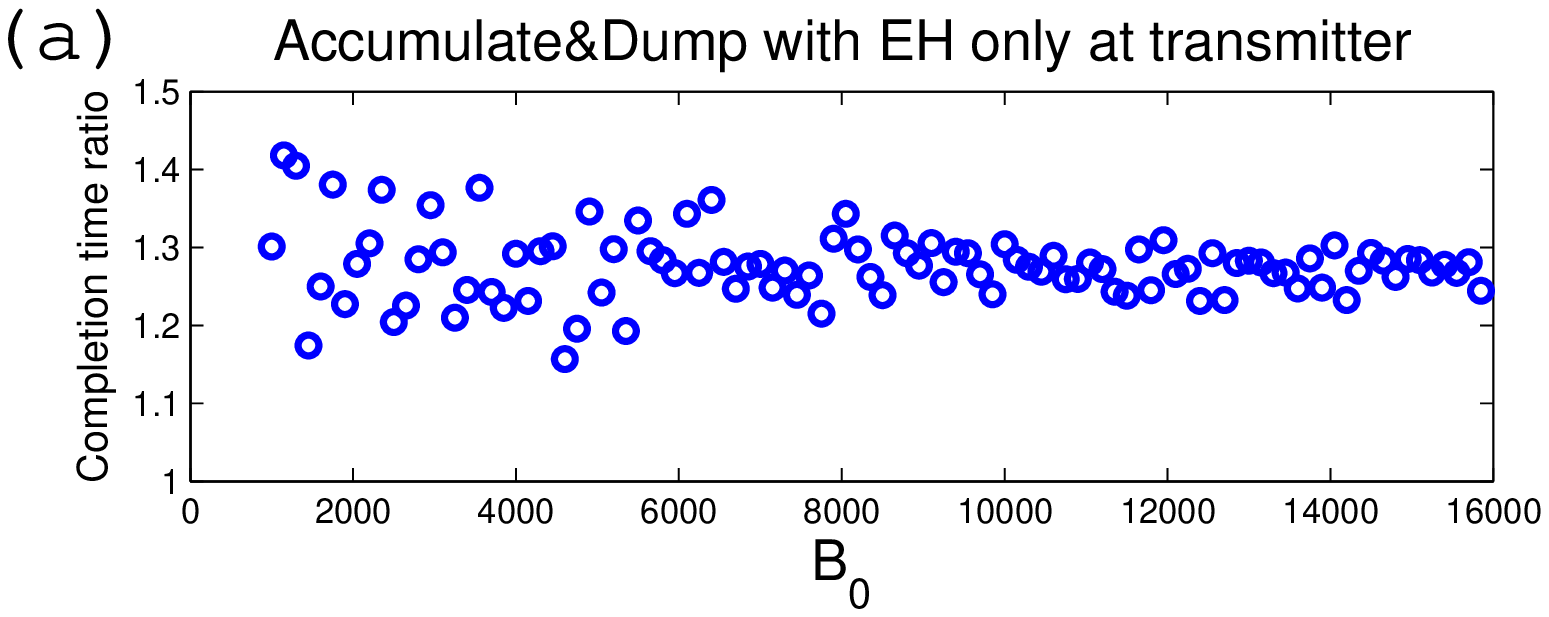}}
	\centerline{\includegraphics[width=8cm]{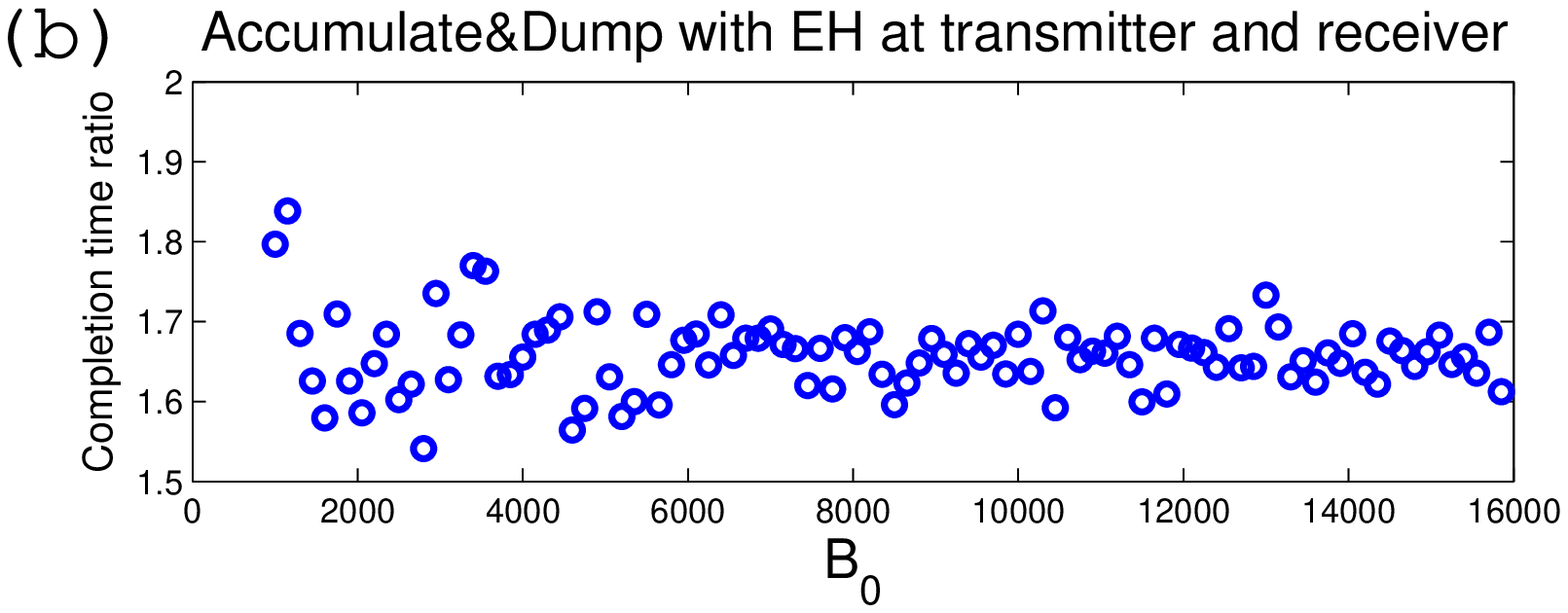}}

	\caption{Comparison of algorithm Accumulate\&Dump with the optimal offline algorithm presented in \cite{Yener2012optbat} with finite battery capacity in the (a) transmitter model and (b) transmitter-receiver model.}
	\label{figure_finitebat_example}

\end{figure}



%% file: ICCsingle.bbl
\begin{thebibliography}{10}
\providecommand{\url}[1]{#1}
\csname url@samestyle\endcsname
\providecommand{\newblock}{\relax}
\providecommand{\bibinfo}[2]{#2}
\providecommand{\BIBentrySTDinterwordspacing}{\spaceskip=0pt\relax}
\providecommand{\BIBentryALTinterwordstretchfactor}{4}
\providecommand{\BIBentryALTinterwordspacing}{\spaceskip=\fontdimen2\font plus
\BIBentryALTinterwordstretchfactor\fontdimen3\font minus
  \fontdimen4\font\relax}
\providecommand{\BIBforeignlanguage}[2]{{%
\expandafter\ifx\csname l@#1\endcsname\relax
\typeout{** WARNING: IEEEtran.bst: No hyphenation pattern has been}%
\typeout{** loaded for the language `#1'. Using the pattern for}%
\typeout{** the default language instead.}%
\else
\language=\csname l@#1\endcsname
\fi
#2}}
\providecommand{\BIBdecl}{\relax}
\BIBdecl

\bibitem{ozel2012achieving}
O.~Ozel and S.~Ulukus, ``Achieving {AWGN} capacity under stochastic energy
  harvesting,'' \emph{{IEEE} Trans. Inf. Theory}, vol.~58, no.~10, pp.
  6471--6483, 2012.

\bibitem{SharmaEH2014}
R.~Rajesh, V.~Sharma, and P.~Viswanath, ``Capacity of {Gaussian} channels with
  energy harvesting and processing cost,'' \emph{{IEEE} Trans. Inf. Theory},
  vol.~60, no.~5, pp. 2563--2575, May 2014.

\bibitem{dong2014near}
Y.~Dong, F.~Farnia, and A.~{\"O}zg{\"u}r, ``Near optimal energy control and
  approximate capacity of energy harvesting communication,'' \emph{arXiv
  preprint arXiv:1405.1156}, 2014.

\bibitem{sinha2012optimal}
A.~Sinha and P.~Chaporkar, ``Optimal power allocation for a renewable energy
  source,'' in \emph{Communications (NCC), 2012 National Conference on}.\hskip
  1em plus 0.5em minus 0.4em\relax IEEE, 2012, pp. 1--5.

\bibitem{UlukusEH2011b}
J.~Yang and S.~Ulukus, ``Optimal packet scheduling in an energy harvesting
  communication system,'' \emph{{IEEE} Trans. Commun.}, vol.~60, no.~1, pp.
  220--230, 2012.

\bibitem{UlukusEH2011c}
O.~Ozel, K.~Tutuncuoglu, J.~Yang, S.~Ulukus, and A.~Yener, ``Transmission with
  energy harvesting nodes in fading wireless channels: Optimal policies,''
  \emph{{IEEE} J. Sel. Areas Commun.}, vol.~29, no.~8, pp. 1732--1743, 2011.

\bibitem{michelusi2012optimal}
N.~Michelusi, K.~Stamatiou, and M.~Zorzi, ``On optimal transmission policies
  for energy harvesting devices,'' in \emph{Information Theory and Applications
  Workshop (ITA), 2012}.\hskip 1em plus 0.5em minus 0.4em\relax IEEE, 2012, pp.
  249--254.

\bibitem{VazeEHICASSP14}
R.~Vaze and K.~Jagannathan, ``Finite-horizon optimal transmission policies for
  energy harvesting sensors,'' in \emph{International Conference on Acoustics,
  Speech, and Signal Processing (ICASSP)}.\hskip 1em plus 0.5em minus
  0.4em\relax IEEE, 2014.

\bibitem{VazeEH2011}
R.~Vaze, ``Competitive ratio analysis of online algorithms to minimize data
  transmission time in energy harvesting communication system,'' in \emph{IEEE
  INFOCOM 2013}, Apr. 2013.

\bibitem{VazeEH2014}
J.~Doshi and R.~Vaze, ``Long term throughput and approximate capacity of
  transmitter-receiver energy harvesting channel with fading,'' in \emph{to
  appear in IEEE ICCS 2014}, Nov. 2014.

\bibitem{Yener2012enernodes}
K.~Tutuncuoglu and A.~Yener, ``Communicating with energy harvesting
  transmitters and receivers,'' in \emph{Information Theory and Applications
  Workshop (ITA), 2012}, Feb 2012, pp. 240--245.

\bibitem{Sharma2010enernodes}
V.~Sharma, U.~Mukherji, V.~Joseph, and S.~Gupta, ``Optimal energy management
  policies for energy harvesting sensor nodes,'' \emph{{IEEE} Trans. Wireless
  Commun.}, vol.~9, no.~4, pp. 1326--1336, April 2010.

\bibitem{Yener2012optbat}
K.~Tutuncuoglu and A.~Yener, ``Optimum transmission policies for battery
  limited energy harvesting nodes,'' \emph{{IEEE} Trans. Wireless Commun.},
  vol.~11, no.~3, pp. 1180--1189, March 2012.

\bibitem{erkal2013optimal}
H.~Erkal, F.~M. Ozcelik, and E.~Uysal-Biyikoglu, ``Optimal offline broadcast
  scheduling with an energy harvesting transmitter,'' \emph{EURASIP Journal on
  Wireless Communications and Networking}, vol. 2013, no.~1, pp. 1--20, 2013.

\bibitem{ozel2012optimal}
O.~Ozel, J.~Yang, and S.~Ulukus, ``Optimal broadcast scheduling for an energy
  harvesting rechargeable transmitter with a finite capacity battery,''
  \emph{{IEEE} Trans. Wireless Commun.}, vol.~11, no.~6, pp. 2193--2203, 2012.

\bibitem{grimmett1985probability}
\BIBentryALTinterwordspacing
G.~Grimmett and D.~Stirzaker, \emph{Probability and random processes}, ser.
  Oxford science publications.\hskip 1em plus 0.5em minus 0.4em\relax Clarendon
  Press, 1985. [Online]. Available:
  \url{http://books.google.co.in/books?id=B5FRAAAAMAAJ}
\BIBentrySTDinterwordspacing

\end{thebibliography}
